\documentclass{IEEEojcsys}

\usepackage[colorlinks,urlcolor=blue,linkcolor=blue,citecolor=blue]{hyperref}

\usepackage{color,array}

\usepackage{graphicx}

\jvol{00}
\jnum{XX}
\paper{1234567}
\pubyear{2021}
\receiveddate{XX September 2023}
\accepteddate{XX October 2023}
\publisheddate{XX November 2021}
\currentdate{XX November 2021}
\doiinfo{OJCSYS.2021.Doi Number}

\usepackage{amsmath}
\usepackage{amssymb,amsfonts}

\usepackage{algorithmic}
\usepackage{siunitx}
\usepackage{bm}
\usepackage{multirow}
\usepackage{mathrsfs}
\usepackage{algorithm}

\DeclareMathOperator*{\argmin}{arg\,min}

\newcommand{\mat}[1]{\ensuremath{\begin{bmatrix} #1 \end{bmatrix}}}				

\newcommand{\ol}[1]{\overline{#1}}
\newcommand{\ul}[1]{\underline{#1}}
\newcommand{\eps}{\epsilon}
\renewcommand{\Pr}[1]{\textnormal{Pr}\left(#1\right)}
\newcommand{\normsq}[1]{\left\lVert #1 \right\rVert^2}
\newcommand{\norm}[1]{\left\lVert #1 \right\rVert}
\newcommand{\norminf}[1]{\left\lVert #1 \right\rVert_{\infty}}

\renewcommand{\text}[1]{\textnormal{#1}}

\usepackage[acronym, style=alttree, toc=true, xindy, nomain, nonumberlist]{glossaries} 

\newglossary{symbols}{sym}{sbl}{List of Symbols}
\newglossary{notation}{not}{nt}{Notation}

\glsaddkey{dimension}{\glsentrytext{\glslabel}}{\glsentryd}{\GLsentryd}{\glsd}{\Glsd}{\GLSd}

\newacronym{MPC}{MPC}{model predictive control}
\newacronym{iss}{ISS}{input-to-state stability}

\newglossaryentry{xk}{type=symbols,
	sort={x},
	dimension={\ensuremath{ \mathbb{R}^{n_\mathrm{x}}}},
	name={\ensuremath{\bm{x}_k}},
	description={State at step $k$}
}

\newglossaryentry{xhk}{type=symbols,
	sort={x},
	dimension={\ensuremath{ \mathbb{R}^{n_\mathrm{x}} }},
	name={\ensuremath{\hat{\bm{x}}_k}},
	description={Measured state at step $k$}
}

\newglossaryentry{xk1}{type=symbols,
	sort={x},
	dimension={\ensuremath{ \mathbb{R}^{n_\mathrm{x}} }},
	name={\ensuremath{\bm{x}_{k+1}}},
	description={State at step $k+1$}
}

\newglossaryentry{xhk1}{type=symbols,
	sort={x},
	dimension={\ensuremath{ \mathbb{R}^{n_\mathrm{x}} }},
	name={\ensuremath{\hat{\bm{x}}_{k+1}}},
	description={Measured state at step $k+1$}
}

\newglossaryentry{uk}{type=symbols,
	sort={u},
	dimension={\ensuremath{ \mathbb{R}^{n_\mathrm{u}} }},
	name={\ensuremath{\bm{u}_k}},
	description={Input at step $k$}
}

\newglossaryentry{yk}{type=symbols,
	sort={y},
	dimension={\ensuremath{ \mathbb{R}^{n_y} }},
	name={\ensuremath{\bm{y}_k}},
	description={Output at step $k$}
}

\newglossaryentry{dk}{type=symbols,
	sort={d},
	dimension={\ensuremath{\mathcal{D}}},
	name={\ensuremath{\bm{d}_k}},
	description={Disturbance at step $k$}
}

\newglossaryentry{Bd}{type=symbols,
	sort={Bd},
	dimension={\ensuremath{ \mathbb{R}^{n\times o} }},
	name={\ensuremath{\bm{E}}},
	description={Measured state at step $k+1$}
}

\newglossaryentry{feasA}{type=symbols,
	sort={feas},
	dimension={\ensuremath{ \mathbb{R}^{N-L}}},
	name={\ensuremath{\mathcal{A}_\text{F}}},
	description={Feasible decision variables}
}

\newglossaryentry{OCP}{type=symbols,
	sort={ocp},
	dimension={OCP has no dimension},
	name={\ensuremath{\mathbb{P}}},
	description={Feasible decision variables}
}

\newglossaryentry{termSet}{type=symbols,
	sort={terminalset},
	dimension={No},
	name={\ensuremath{\mathcal{Z}}_{\text{f}}},
	description={TerminalSet}
}

\newglossaryentry{termSetNoise}{type=symbols,
	sort={terminalset},
	dimension={No},
	name={\ensuremath{\hat{\mathcal{Z}}_{\text{f}}}},
	description={TerminalSet}
}

\newglossaryentry{tightX}{type=symbols,
	sort={tightX},
	dimension={\ensuremath{ \mathbb{R}^{n_\mathrm{x}}}},
	name={\ensuremath{\mathcal{Z}}},
	description={TerminalSet}
}

\newglossaryentry{X}{type=symbols,
	sort={X},
	dimension={\ensuremath{ \mathbb{R}^{n_\mathrm{x}}}},
	name={\ensuremath{\mathcal{X}}},
	description={TerminalSet}
}

\newglossaryentry{U}{type=symbols,
	sort={U},
	dimension={\ensuremath{ \mathbb{R}^{n_\mathrm{u}}}},
	name={\ensuremath{\mathcal{U}}},
	description={InputConstraints}
}

\newglossaryentry{tightXnoise}{type=symbols,
	sort={tightX},
	dimension={No},
	name={\ensuremath{\hat{\mathcal{Z}}}},
	description={TerminalSet}
}

\newglossaryentry{numSamp}{type=symbols,
	sort={numSamp},
	dimension={No},
	name={\ensuremath{N_\text{S}}},
	description={Number of samples}
}

\newglossaryentry{numDiscSamp}{type=symbols,
	sort={numSamp},
	dimension={No},
	name={\ensuremath{N_\text{D}}},
	description={Number of samples}
}

\newglossaryentry{tightU}{type=symbols,
	sort={tightneedU},
	dimension={No},
	name={\ensuremath{\hat{\mathcal{U}}_l}},
	description={Feasible decision variables}
}

\newglossaryentry{firstSet}{type=symbols,
	sort={firststepconstraint},
	dimension={No},
	name={\ensuremath{\hat{\mathcal{Z}}_{I}}}, 
	description={Feasible decision variables}
}

\newglossaryentry{wk}{type=symbols,
	sort={w},
	dimension={\ensuremath{ \mathcal{W} }},
	name={\ensuremath{\bm{w}_k}},
	description={General / extended disturbance at step $k$}
}

\newglossaryentry{w}{type=symbols,
	sort={w1},
	dimension={\ensuremath{ \mathbb{R}^{n_\mathrm{w}} }},
	name={\ensuremath{\bm{w}}},
	description={General / extended disturbance}
}

\newglossaryentry{ulaw}{type=symbols,
	sort={k},
	dimension={\ensuremath{ \mathbb{R}^{n_\mathrm{u}} }},
	name={\ensuremath{\bm{\kappa}}},
	description={General / extended disturbance}
}

\newglossaryentry{noise}{type=symbols,
	sort={es},
	dimension={\ensuremath{\mathcal{M}}},
	name={\ensuremath{\bm{\mu}}},
	description={Measurement noise was epsilon before}
}

\newglossaryentry{epsk}{type=symbols,
	sort={e},
	dimension={\ensuremath{\mathcal{M}}},
	name={\ensuremath{\bm{\mu}_k}},
	description={Measurement noise at step $k$}
}

\newglossaryentry{epsk1}{type=symbols,
	sort={e},
	dimension={\ensuremath{\mathrm{THISISWRONGCHANGEIT}}},
	name={\ensuremath{\bm{\mu}_{k+1}}},
	description={Measurement noise at step $k+1$}
}

\newglossaryentry{epsl}{type=symbols,
	sort={e},
	dimension={\ensuremath{ \mathcal{E}_{\mu,l}}},
	name={\ensuremath{\bm{A}\bm{\eps}_{l|k}}},
	description={Set for the influence of the measurement noise on the predicted step $l$}
}

\newglossaryentry{epsl1}{type=symbols,
	sort={e},
	dimension={\ensuremath{ \mathcal{E}_{\mu,1}}},
	name={\ensuremath{\bm{A}\bm{\eps}_{1|k}}},
	description={Set for the influence of the measurement noise on the predicted step $l=1$}
}

\newglossaryentry{epsL}{type=symbols,
	sort={e},
	dimension={\ensuremath{ \mathcal{E}_{\mu,L}}},
	name={\ensuremath{\bm{A}\bm{\eps}_{L|k}}},
	description={Set for the influence of the measurement noise on the predicted step $l=L$}
}

\newglossaryentry{udata}{type=symbols,
	sort={u},
	dimension={\ensuremath{ \mathbb{R}^{n_\mathrm{u}} }},
	name={\ensuremath{\bm{u}^{\textnormal{d}}}},
	description={Input data}
}

\newglossaryentry{ddata}{type=symbols,
	sort={d},
	dimension={\ensuremath{ \mathbb{R}^{n_\mathrm{d}} }},
	name={\ensuremath{\bm{d}^{\textnormal{d}}}},
	description={Disturbance data}
}

\newglossaryentry{xdata}{type=symbols,
	sort={x},
	dimension={\ensuremath{ \mathbb{R}^{n_\mathrm{x}} }},
	name={\ensuremath{\bm{x}^{\textnormal{d}}}},
	description={State data}
}

\newglossaryentry{roa}{type=symbols,
	sort={X},
	name={\ensuremath{\mathcal{X}_0}},
	description={Region of attraction for the state}
}

\usepackage{amsthm}

\newtheorem{theorem}{Theorem}
\newtheorem{lemma}{Lemma}
\newtheorem{proposition}{Proposition}
\newtheorem{remark}{Remark}
\newtheorem{corollary}{Corollary}
\newtheorem{definition}{Definition}
\newtheorem{assumption}{Assumption}

\begin{document}

\sptitle{Article Category}

\title{Data-driven Tube-Based Stochastic Predictive Control} 
\editor{This paper was recommended by Associate Editor F. A. Author.}

\author{S. KERZ\affilmark{1}}

\author{J. TEUTSCH\affilmark{1}}

\author{T. BR\"UDIGAM\affilmark{1}}

\author{M. LEIBOLD\affilmark{1}}

\author{D. WOLLHERR\affilmark{1}}

\affil{Technical University Munich, Munich, Germany} 

\corresp{CORRESPONDING AUTHOR: S. Kerz (e-mail: \href{mailto:s.kerz@tum.de}{s.kerz@tum.de})}

\markboth{PREPARATION OF PAPERS FOR IEEE OPEN JOURNAL OF CONTROL SYSTEMS}{F. A. AUTHOR {\itshape ET AL}.}

\begin{abstract}
A powerful result from behavioral systems theory known as the fundamental lemma allows for predictive control akin to Model Predictive Control (MPC) for linear time invariant (LTI) systems with unknown dynamics purely from data. While most of data-driven predictive control literature focuses on robustness with respect to measurement noise, only few works consider exploiting probabilistic information of disturbances for performance-oriented control as in stochastic MPC.
In this work, we propose a novel data-driven stochastic predictive control scheme for chance-constrained LTI systems subject to measurement noise and additive stochastic disturbances.
In order to render the otherwise stochastic and intractable optimal control problem deterministic, our approach leverages ideas from tube-based MPC by decomposing the state into a deterministic nominal state driven by inputs and a stochastic error state affected by disturbances.
By tightening constraints probabilistically with respect to the additive disturbance and robustly with respect to the measurement noise, satisfaction of original chance-constraints is guaranteed.
The resulting data-driven receding horizon optimal control problem is lightweight and recursively feasible, and renders the closed loop input-to-state stable with respect to both additive disturbances and measurement noise. 
We demonstrate the effectiveness of the proposed approach in a simulation example.
\end{abstract}


\maketitle

\section{INTRODUCTION}\label{sec:introduction}
As sensor data and computational power are becoming more widely available, data-driven approaches are increasingly relevant in modern control applications. 
Recently, direct data-driven control design based on Willems' lemma~\cite{willems2005note} has emerged as an appealing alternative to the classical approach of first constructing a model from data based on system identification methods. 
Informally, the so-called \emph{fundamental lemma} states that for discrete-time linear time-invariant (LTI) systems, the time-shifted vectors of any input-output trajectory generated by a persistently exciting input signal span the vector space of all (fixed length) input-output trajectories of the system. As a consequence, discrete-time LTI systems can be represented by a single measured trajectory, enabling control and analysis problems to be solved directly from trajectory data~\cite{markovsky2008,de2019formulas,datainformativity,berberichECC2020,Yin2021MaxLike}. A concise and comprehensive recent review is provided in~\cite{behavioraltheory2021}. 
By replacing the model with trajectory data, one effectively obtains a non-parametric representation of the subspace spanning the system behavior. This behavioral subspace can be directly searched by varying the coefficients of the linear combination of the basis (or \emph{library of trajectories}~\cite{behavioraltheory2021}), making the framework naturally well suited for finding optimal future input-output sequences within data-driven (or \emph{data-enabled}) predictive control~\cite{ACC2015,coulson2019data}.
Data-driven predictive controllers have since been modified to apply for different classes of systems~\cite{faulwasser2022journey,huang.KernelDeepC,Berberich.2022} and give theoretical guarantees for various settings.
Several works consider data perturbed by measurement noise and provide robust extensions of data-driven predictive control, akin to robust MPC (RMPC)~\cite{berberich2020data, berberich2020TrackingMPC,berberich2020robust,berberich2021design,bongard2022stabAnal,Huang.2023}.
Few works however consider a setting akin to SMPC, i.e., avoid overly conservative performance by exploiting probabilstic knowledge of stochastic disturbances.
\subsubsection*{Related work}
Systems subject to stochastic disturbances are considered in~\cite{coulson2019regularized,coulson2021distributionally} for a very general setting in which no knowledge about the probability distribution of the disturbance is assumed and the data available are inexact. Strong probabilistic guarantees on out-of-sample performance are given for the resulting open-loop data-driven optimal control problem, whereas probabilistic constraint satisfaction is enforced for the worst-case probability distribution that would explain the data. However, closed-loop properties and recursive feasibility of the optimal control problem are not considered. 
In~\cite{wang2022InnovationEstimation}, the authors present a data-driven predictive control scheme for systems affected by zero-mean white Gaussian process and measurement noise. By extending Willems' fundamental lemma to incorporate innovation data, the innovation form of the underlying state space system is represented from data. If the innovation data are exact, the resulting predictor is optimal, but ensuring stability and robustness of the closed-loop system remains an open research challenge.

In literature, only one line of work~\cite{pan2022StochasticLemma,pan2022ddsmpc,Pan2022DDoutputSMPC} falls into the category of data-driven SMPC, i.e., utilizes distributional knowledge of uncertainties for performance-oriented control based on chance-constraints, and comes with certificates for recursive feasibility and stability.  
In~\cite{pan2022StochasticLemma}, Pan et al. presented an extension of the fundamental lemma for stochastic LTI systems that leverages polynomial chaos expansion (PCE). Due to the linearity of PCE coefficients, stochastic variables described by those coefficients can be propagated through the dynamics in full. Applied to systems subject to stochastic additive disturbances with known probability distribution, this allows for a deterministic reformulation of the otherwise intractable stochastic optimal control problem. The resulting data-driven SMPC scheme is rendered recursively feasible and practically stable with backup initial conditions~\cite{pan2022ddsmpc} and is extended to the input-output case in~\cite{Pan2022DDoutputSMPC}.
Since distributional knowledge informs predictions in full, the approach is very appealing for settings in which the distribution of the disturbance is known exactly, but comes at the cost of relatively heavy online computations. 

\subsubsection*{Contribution}
In this work, we present a novel data-driven SMPC scheme for the performance-oriented control of unknown LTI systems subject to stochastic disturbances based on trajectory data. 
By employing ideas from tube-based MPC, we provide a deterministic reformulation of the stochastic optimal control problem that is lightweight, recursively feasible and leads to a stable closed loop.
Similarly to~\cite{pan2022ddsmpc}, we assume that disturbance realizations can be measured or estimated offline, before the control phase. In contrast to~\cite{pan2022ddsmpc}, we consider a setting in which no further distributional information is available, and an additional uncertainty, in the form of online state measurement noise, needs to be accounted for.

In order to combine the deterministic fundamental lemma with stochastic dynamics, we split the state into a deterministic nominal part affected by inputs and a stochastic error part affected by disturbances. By formulating the optimal control problem with respect to the nominal state, the optimization problem is rendered deterministic. Offline, before the actual control phase, probabilistic error predictions are used to appropriately tighten constraints on the nominal state such that original state constraints are met with pre-specified probability level. 
We represent the decomposed dynamics in the data-driven setting, and show how both the error predictions for the probabilistc constraint tightening and an additional robust constraint tightening with respect to bounded online measurement noise can be computed from data.
To render the closed loop stable and the optimal control problem recursively feasible, we employ ideas from classical SMPC in the design of a data-driven robust first step constraint~\cite{lorenzen2017stochastic}, stochastic tubes~\cite{mesbah2016} with tightened constraint sets~\cite{lorenzen2016constraint}, terminal ingredients~\cite{goulart2006optimization,RawlingsMayneDiehl2017}, and a pre-stabilizing data-driven state feedback, with gains derived from the initially measured trajectory~\cite{de2019formulas}.

The main contribution is a recursively feasible and stable tube-based data-driven predictive control scheme that
leverages chance constraints to efficiently control against process disturbances in the presence of measurement noise.

\subsubsection*{Outline}
The remainder of this work is structured as follows. Section~\ref{sec:Problem} states the problem setting and all standing assumptions. Section~\ref{sec:Prelim} introduces preliminaries on tube-based MPC. Section~\ref{sec:DDSMPC} proposes the data-driven SMPC algorithm, and Section~\ref{sec:properties} presents ingredients for closed loop certificates. Section~\ref{sec:Sim} tests the algorithm in simulation. Section~\ref{sec:Discussion} discusses the case of noisy offline data, and we conclude the paper in Section~\ref{sec:conclusion}.
\subsection*{NOTATION}\label{sec:Notation}
We write boldface uppercase (resp., lowercase) letters to denote matrices (resp., vectors) and $\bm{0}$ for any zero matrix or vector. $\bm{I}_n$ is the $n \times n$ identity matrix and $\mathbb{N}_a^b$ abbreviates the integer sequence $\left\{a,\ldots, b\right\}$. $\bm{A}\succ 0$ ($\bm{A}\succeq 0$) means matrix $\bm{A}$ is positive (semi-)definite. Sequences of vectors are written as $\bm{s}_{[1,\,N]} = \left(\bm{s}_1, \dots, \bm{s}_N\right)$, whereas $\underline{\bm{s}}_{[1,\,N]} = \left[\bm{s}_1^{\top}, \dots, \bm{s}_N^{\top}\right]^{\top}$ denotes the column vector that results from stacking all vectors of the sequence. The weighted 2-norm $\sqrt{\bm{x}^\top\bm{Q}\bm{x}}$ of $\bm{x}$ is abbreviated by $\norm{\bm{x}}_{\bm{Q}}$. For a sequence $\bm{s}_{[1,\,N]}$, let $\norm{\bm{s}_{[1,\,N]}}_\infty=\sup\{|\bm{s}_1|,\ldots,|\bm{s}_N|\}$, and $\bm{s}_{[1,\,N]}\subset\mathcal{S}$ denote that $\bm{s}_i\in\mathcal{S}$ for all $i\in\mathbb{N}_1^N$.
The probability of an event $X$ is denoted by $\Pr{X}$.
For any sequence of vectors $\bm{s}_{[0,\,N]}$, the corresponding Hankel matrix $\bm{H}_{L}\left(\bm{s}_{[0,\,N]}\right)$ is
\begin{equation} \label{def:hankel}
\bm{H}_{L}\left(\bm{s}_{[0,\,N]}\right) = \mat{\bm{s}_0 & \bm{s}_1 & \cdots & \bm{s}_{N-L+1} \\ 
	\bm{s}_1 & \bm{s}_2 & \cdots & \bm{s}_{N-L+2} \\
	\vdots & \vdots & \ddots & \vdots \\
	\bm{s}_{L-1} & \bm{s}_{L} & \cdots & \bm{s}_{N}}.
\end{equation}
For a matrix $\bm{H}$ and appropriate integers $a,b$, $\left[\bm{H}\right]_{a}$ is the $a$-th row vector of $\bm{H}$ and $\left[\bm{H}\right]_{[a,b]}$ represents the submatrix of $\bm{H}$ that is comprised of rows $a,\ldots,b$.
For any sets $\mathcal{A},\mathcal{B}\subset\mathbb{R}^n$, we write the Minkowski set addition as $\mathcal{A} \oplus \mathcal{B}=\{a+b \mid a \in \mathcal{A}, b \in \mathcal{B}\}$, the Pontryagin set difference as $\mathcal{A} \ominus \mathcal{B} = \{a \in \mathcal{A} \mid (\forall  b \in \mathcal{B})\,a + b \in \mathcal{A}  \}$ and set multiplication as $\bm{K}\mathcal{A} = \{\bm{K} a \mid a \in \mathcal{A}\}$.
In the context of predictive control, we write $\bm{x}_{l|k}$ for the predicted state $l$ steps ahead of $\bm{x}_k$.
A continuous function $\gamma: \mathbb{R}_0^+ \rightarrow \mathbb{R}_0^+$ is of class $\mathscr{K}$ if $\gamma$ is strictly increasing and $\gamma(0) = 0$.
    If $\gamma\in\mathscr{K}$ and $\gamma$ is unbounded, then $\gamma\in\mathscr{K}_{\infty}$.
    If $\gamma: \mathbb{R}_0^+ \rightarrow \mathbb{R}_0^+$ is continous, strictly decreasing and ${\lim\limits_{t \rightarrow \infty}\gamma(t) = 0}$, then it is of class $\mathscr{L}$.
    If $\delta\in \mathbb{R}_0^+ \times \mathbb{R}_0^+ \rightarrow \mathbb{R}_0^+$ is such that $\delta(\cdot,t) \in \mathscr{K}$ and $\delta(r,\cdot) \in \mathscr{L}$,  $\delta$ is of class $\mathscr{KL}$.

\section{PROBLEM SETUP}\label{sec:Problem}
We consider a linear, discrete-time time-invariant system
\begin{subequations}\label{eq:system}
	\begin{align}
	\gls*{xk1} &= \bm{A} \gls*{xk} + \bm{B} \gls*{uk} + \gls*{Bd} \gls*{dk}, \label{eq:system1}\\
	\gls*{xhk} &= \gls*{xk} + \gls*{epsk}\label{eq:system2}
	\end{align}
\end{subequations}
with state $\gls*{xk}\in \glsd*{xk}$, input $\gls*{uk}\in \glsd*{uk}$ and controllable pairs $(\bm{A},\bm{B})$, $(\bm{A},\gls*{Bd})$.
System~\eqref{eq:system} is subject to two kinds of uncertainty: 
A disturbance $\gls*{dk}\in \glsd*{ddata}$ in the dynamics, akin to process noise, and a state measurement error $\gls*{epsk} \in \glsd*{xk}$, akin to measurement noise.
The problem is to design a data-driven stochastic predictive control algorithm that stabilizes the origin 
while respecting input and state constraints. 
The predictive control algorithm centers around repeated solutions of a stochastic optimal control problem (OCP)
\begin{subequations} \label{eq:ocp_intro}
	\begin{align}
		\underset{\bm{U}_k}{\text{minimize}}~&\mathbb{E}\left(\sum\limits_{l = 0}^{L-1} 
		J_\text{s}(\bm{x}_{l|k},\bm{u}_{l|k})+J_\text{f}(\bm{x}_{L|k})\right)\label{eq:ocp_intro_cost}\\	
		\text{s.t.}~ &\bm{x}_{0|k} = \bm{x}_k, \\		
		&\bm{x}_{l|k}~\text{evolves according to \eqref{eq:system1}},\label{ocpintro_evolve}\\
		& \bm{u}_{l|k} \in \gls*{U}\qquad\forall l\in\mathbb{N}_0^{L-1},\label{eq:ocp_intro_InputCon}\\
		& \Pr{\bm{x}_{l|k} \in \gls*{X}}\geq p\qquad\forall l\in\mathbb{N}_0^{L-1},\label{eq:ocp_intro_chancecon}
	\end{align}
\end{subequations}
where $\mathbb{E}(\cdot)$ denotes the expected value, and $J_\text{s}$ and $J_\text{f}$ represent stage and terminal costs, respectively. State constraints must be satisfied in probability, as specified by the chance constraint~\eqref{eq:ocp_intro_chancecon} with user-specified risk parameter $p\in (0,1]$.
Both $\gls*{U}$ and $\gls*{X}$ are user-specified convex polytopic sets
\begin{subequations}
\begin{align}
\gls*{U} &= \left\{\bm{u} \in \glsd*{uk} ~\left|~ \bm{G}_u \bm{u} \le \bm{g}_u \right.\right\}, \label{eq:inputcons}\\
\gls*{X} &= \left\{\bm{x} \in \glsd*{xk} \hspace{2.5pt} ~\left|~ \bm{G}_x \bm{x} \le \bm{g}_x \right.\right\}\label{eq:statecons}.
\end{align}
\end{subequations}
Problem~\eqref{eq:ocp_intro} is solved in a receding horizon fashion, where in each time step $k$, the first input $\bm{u}_{0|k}$ of the minimizing solution $\bm{U}_k=\left(\bm{u}_{0|k},\ldots,\bm{u}_{L-1|k}\right)$ is applied to the actual system~\eqref{eq:system}.
Due to the probabilistic nature of both the system evolution~\eqref{ocpintro_evolve} and the chance constraints~\eqref{eq:ocp_intro_chancecon}, problem~\eqref{eq:ocp_intro} is intractable. One aim of this work is to provide a tractable reformulation based on data.
\subsection{Available data}
The system matrices $\bm{A}$, $\bm{B}$, and $\gls*{Bd}$ of system~\eqref{eq:system1} are unknown. In their stead, we use trajectory data to design the predictive control scheme.
\begin{assumption}\label{assum:trajData}
A trajectory $(\gls*{udata}_{[0,N-1]},\,\gls*{ddata}_{[0,N-1]},\,\gls*{xdata}_{[0,N-1]})$ generated by~\eqref{eq:system1} is available, where both $\gls*{udata}_{[0,N-1]}$ and $\gls*{ddata}_{[0,N-1]}$ are persistently exciting of order $L+n_\mathrm{x}+1$. 
\end{assumption}
\begin{definition}[Persistency of excitation] \label{def:persistency} 
	A sequence of vectors {$\bm{s}_{[0,N-1]} \in \mathbb{R}^m$} is persistently exciting of order $L$ if the Hankel matrix $\bm{H}_{L}\left(\bm{s}_{[0,N-1]}\right)$ has full row rank $mL$.
\end{definition}
\begin{remark}
For ease of exposition, and as in related works~\cite{pan2022ddsmpc,wang2022InnovationEstimation}, we consider full state measurements, instead of output measurements. An extension to the input-output setting can follow from considering an extended state of past outputs, as in for example~\cite{berberich2021design,Pan2022DDoutputSMPC}.

If disturbance recordings cannot be accessed, they may be estimated from input-state data as in~\cite{pan2022ddsmpc}. The availability of noise-free offline data is restrictive. We discuss the influence of noisy data in Section~\ref{sec:Discussion} and show its effect in a simulation example in Section~\ref{sec:Sim}.
\end{remark}
In order to guarantee the satisfaction of state constraints~\eqref{eq:ocp_intro_chancecon}, we employ a stochastic constraint tightening with respect to the disturbance $\bm{d}$, and a robust constraint tightening with respect to the measurement noise $\gls*{noise}$. 
To this end, we make the following assumption.
\begin{assumption}\label{assum:DisturbanceBounds}
The following statements hold for disturbance $\bm{d}$ and measurement noise $\gls*{noise}$:\\
    (a) In addition to the data in Assumption~\ref{assum:trajData}, $\gls*{numSamp}$ samples of disturbance sequences $\bm{d}^{(i)}_{[0,\,L-1]}$, $i=0,\ldots \gls*{numSamp}$ are available.\\
    (b) The noise $\gls*{noise}$ is bounded by a known polytopic set
        \begin{equation}
            \glsd*{epsk} = \left\{\gls*{noise} \in \glsd*{xk} ~\left|~ \bm{G}_{\mu} \gls*{noise} \le \bm{g}_{\mu} \right.\right\}.
        \end{equation}
        that contains the origin.\\
    (c) The disturbance $\bm{d}$ is bounded by a known polytopic set
        \begin{equation}\label{eq:distset}
            \glsd*{dk} = \left\{ \bm{d} \in \glsd*{ddata} ~\left|~ \bm{G}_d \bm{d} \le \bm{g}_d \right.\right\}
        \end{equation}
        that contains the origin.
\end{assumption}
Items (a) and (b) of Assumption~\ref{assum:DisturbanceBounds} are necessary by virtue of stochastic and robust constraint handling, respectively. The addition of item (c) allows for recursive feasibility and stability in a robust sense, i.e., both properties can only be guaranteed in probability if either $\bm{d}$ or $\gls*{noise}$ are unbounded. 
\subsection{Desired properties}\label{sec:iss_prelim}
We want to guarantee that the OCP remains solvable while the system is steered by the predictive controller. Given that the OCP is feasible at initial time, \emph{recursive feasibility} guarantees that constraint satisfaction is possible for all time.
\begin{definition}[Recursive Feasibility]\label{def:recFeasPrelim}
    A receding horizon OCP is \emph{recursively feasible} if the existence of an admissible solution $\bm{U}_k$ implies the existence of an admissible solution $\bm{U}_{k+1}$ at the next time step.
\end{definition}
Additionally, we desire input-to-state stability of the closed-loop system
\begin{equation} \label{eq:sys_iss}
	\gls*{xk1} = \bm{f}\left(\gls*{xk}, \gls*{wk}\right),
\end{equation}
with respect to $\gls*{wk} \in  \glsd*{w}$, which is an extended disturbance term comprising all uncertainty (both from the measurement noise and the additive disturbance). 
For the following definitions, let the extended disturbance be bounded by a compact support set $\glsd*{wk}$ that contains the origin. Furthermore, let the origin be an equilibrium of the closed loop system for zero disturbance, $\bm{f}\left(\bm{0}, \bm{0}\right) = \bm{0}$, as is the case for system~\eqref{eq:system1}. Denote by $\bm{\phi}\left(k,\bm{x}_0, \bm{w}_{[0,k]}\right)$ the solution to~\eqref{eq:sys_iss} at time step $k$ for initial state $\bm{x}_0$ and disturbance sequence $\bm{w}_{[0,k]}$.
\begin{definition}[Robust positive invariant set]
A set $\gls*{roa}$ is \emph{robust positive invariant} for system~\eqref{eq:sys_iss} if
$\bm{x}_0\in\gls*{roa} \implies \bm{\phi}\left(k,\bm{x}_0, \bm{w}_{[0,k]}\right)\in\gls*{roa}$
for all $k\in\mathbb{N}$, and for all disturbance realizations $\bm{w}_{[0,k]}\subset\glsd*{wk}$.
\end{definition}
\begin{definition}[Input-to-state stability {~\cite[Def. 19]{goulart2006optimization}}] \label{def:iss} 
Let $\gls*{roa}\subseteq\glsd*{xk}$ be a closed robust positive invariant set for system \eqref{eq:sys_iss} with the origin in its interior. System \eqref{eq:sys_iss} is \emph{input-to-state stable (ISS)} with respect to the disturbance and with region of attraction $\gls*{roa}$, if there exist functions $\beta\in\mathscr{KL}$ and $\gamma\in\mathscr{K}$ such that for all $k \in \mathbb{N}_0$, $\bm{x}_0 \in \gls*{roa}$, and $\gls*{wk}\in\glsd*{wk}$
\begin{equation} \label{eq:iss_cond}
	\norm{\bm{\phi}\left(k,\bm{x}_0, \bm{w}_{[0,k]}\right)} \le \beta\left(\norm{\bm{x}_0},k\right) + \gamma\left(\norm{\bm{w}_{[0,k]}}_\infty\right).
\end{equation}
\end{definition}
By causality, the same definition would result if $\bm{w}_{[0,k]}$ is replaced by $\bm{w}_{[0,k-1]}$~\cite{Roset.2008}.
Definition \ref{def:iss} implies that the origin of the undisturbed system $\gls*{xk1} = \bm{f}\left(\gls*{xk}, \bm{0}\right)$ is asymptotically stable with region of attraction $\gls*{roa}$.
It is equivalent to \emph{robust asymptotic stability} of the origin as defined in~\cite[Definition 4.43]{RawlingsMayneDiehl2017}.

\section{TUBE-BASED MPC}\label{sec:smpc}\label{sec:Prelim}
In this work, we present a tractable deterministic reformulation of the stochastic OCP~\eqref{eq:ocp_intro} that is based on trajectory data. This reformulation uses ideas from stochastic and robust tube-based MPC~\cite{mesbah2016,RawlingsMayneDiehl2017}, which are introduced in the following. 
In tube-based MPC, a key idea is to decompose the state $\bm{x}$ into a deterministic nominal part $\bm{z}$ and a (possibly) stochastic error part $\bm{e}$,
\begin{equation}
    \bm{x}_{k} = \bm{z}_{k} + \bm{e}_{k}\label{eq:state_decomp}.
\end{equation}
Given the underlying system~\eqref{eq:system1}, the resulting dynamics are
\begin{subequations}
  \begin{align}
    \bm{z}_{k+1} &= \bm{A}\bm{z}_{k} + \bm{B} \bm{u}_{k},\quad\bm{x}_0=\bm{z}_0,\label{eq:nominal_dynamics}\\
    \bm{e}_{k+1} &= \bm{A}\bm{e}_{k} + \gls*{Bd} \bm{d}_{k},\quad\,\bm{e}_0=\bm{0}.\label{eq:error_dynamics}
\end{align}   
\end{subequations}
By only considering the deterministic nominal dynamics~\eqref{eq:nominal_dynamics} in the predictive controller, the OCP~\eqref{eq:ocp_intro} is rendered deterministic. The expected value in the cost function~\eqref{eq:ocp_intro_cost} can then be computed explicitly, and all terms depending on the state error can be neglected, since they are constant.

For any finite time step $k$, the error $\bm{e}$ is bounded and the actual state $\bm{x}$ is in a neighborhood of the nominal state $\bm{z}$. For an $L$-step trajectory, these neighborhoods form a tube around the nominal state predictions.
The original state $\bm{x}$ satisfies user-specified constraints $\gls*{X}$ if the constraints on the nominal state $\bm{z}$ are tightened with respect to that tube.
Let us denote the tightened constraint sets for the nominal states by $\gls*{tightX}_l$, i.e., $\bm{z}_{l|k}\in\gls*{tightX}_l$ replaces \eqref{eq:ocp_intro_chancecon} and $\{\gls*{tightX}_0,\ldots,\gls*{tightX}_{L-1}\}$ describes the nominal state tube.

A robust constraint tightening (as in RMPC) computes $\gls*{tightX}_l$ such that for all possible realizations of uncertainty
\begin{equation}\label{eq:RobustTubeImplies}
    \bm{z}_{l|k}\in\gls*{tightX}_l\implies \bm{x}_{l|k}\in\gls*{X}.
\end{equation}
Naturally, $\gls*{tightX}_l$ should be as large as possible, as every other choice introduces unnecessary conservatism. The maximal sets that satisfy~\eqref{eq:RobustTubeImplies} are
$\gls*{tightX}_l=\gls*{X}\ominus \mathcal{E}_l$,
where $\mathcal{E}_l=\{\bm{e}\in\glsd*{xk}\mid \left(\exists\bm{d}_0,\ldots,\bm{d}_{l-1}\in\mathcal{D}\right)~ \bm{e}=\bm{e}_l~\text{from~\eqref{eq:error_dynamics}}\}$ encompasses the resulting state errors $\bm{e}_{l|k}$ for all possible realizations of the disturbance.

In a probabilistic constraint tightening (as in SMPC), the state constraints~\eqref{eq:statecons} are required to hold with probability level $p$ for each future predicted state $\bm{x}_{l|k}$, that is
\begin{equation}\label{eq:StochasticTubeImplies}
    \bm{z}_{l|k}\in\gls*{tightX}_l\implies \Pr{\bm{x}_{l|k}\in\gls*{X}}>p,
\end{equation}
where the conditional dependency on $\bm{x}_{0|k} = \bm{x}_{k}$ is understood and omitted in the following.
With $\gls*{X}$ defined as in \eqref{eq:statecons} and the state decomposed as in~\eqref{eq:state_decomp}, the right-hand-side of~\eqref{eq:StochasticTubeImplies} 
can be split into two separate expressions
\begin{subequations}\label{twoConsSMPC}
\begin{align}
    &\bm{G}_{x} \bm{z}_{l|k} \le \tilde{\bm{\eta}},\label{eq:detCon}\\
    &\Pr{\tilde{\bm{\eta}} \le \bm{g}_{x} - \bm{G}_{x} \bm{e}_{l|k} } \ge p,\label{eq:probCon}
\end{align}
\end{subequations}
with the introduction of a new parameter $\tilde{\bm{\eta}} \in \mathbb{R}^{r_x}$.
The deterministic constraint~\eqref{eq:detCon} is used online, i.e., replaces the chance constraints~\eqref{eq:ocp_intro_chancecon} in the optimal control problem. 
Offline, before the actual control phase, the parameter $\tilde{\bm{\eta}}$ is computed such that~\eqref{eq:probCon} holds. Tightened nominal constraints are defined as
\begin{equation} \label{eq:stateprobcons_tight}
\gls*{tightX}_l = \left\{\bm{z} \in \glsd*{xk}  ~\left|\right.~ \bm{G}_x \bm{z} \le \bm{\eta}_l\right\},
\end{equation}
and to minimize conservatism, $\bm{\eta}_l$ is determined by solving
\begin{subequations} \label{eq:stateprobcons_eta}
	\begin{align}
	\bm{\eta}_l &= \max\limits_{\tilde{\bm{\eta}}} \tilde{\bm{\eta}} \\
	&\text{s.t. } \Pr{\tilde{\bm{\eta}} \le \bm{g}_{x} - \bm{G}_{x} \bm{e}_{l} } \ge p,\label{eq:ChanceOptimConstraint}
	\end{align}
\end{subequations}
with $\max(\cdot)$ applied element-wise.
If the true disturbance distribution is unknown, and a finite number of $\gls*{numSamp}$ disturbance sequences $\bm{d}_{[0,l-1]}$ are available, the chance-constrained optimization problem~\eqref{eq:stateprobcons_eta} may be solved approximately, by reformulating it based on computed state error samples $\bm{e}_{l}$.
To that end, a multitude of methods are available, see the recent survey~\cite{Kucukyavuz.2022}.
For simplicity, we employ a scenario approximation~\cite{Campi.2011}, in which~\eqref{eq:stateprobcons_eta} is reformulated into a large-scale linear program where the chance constraint~\eqref{eq:ChanceOptimConstraint} is replaced by deterministic constraints
\begin{equation}\label{eq:DetEta}
    \tilde{\bm{\eta}} \le \bm{g}_{x} - \bm{G}_{x} \bm{e}_{l}
\end{equation}
and required to hold for all but $\gls*{numDiscSamp}$, $\gls*{numDiscSamp}<\gls*{numSamp}$, samples, which are discarded. A detailed description of the sampling-based solution to~\eqref{eq:stateprobcons_eta} is provided in the supplementary material.
\begin{remark}
If the distribution of the disturbance is known, the chance-constrained optimization problem~\eqref{eq:stateprobcons_eta} may be solved analytically. 
The analytic solution requires propagating the probability distribution of the disturbance through the error dynamics~\eqref{eq:error_dynamics}, which is nontrivial in general as it requires the convolution of distributions. In the data-driven setting, the stochastic fundamental lemma~\cite{pan2022StochasticLemma} may be employed.
\end{remark}
In Section~\ref{sec:DDSMPC}-\ref{sec:tightcons}, we show how state error samples $\bm{e}^{(i)}_{l}$ are computed from disturbance samples $\bm{d}^{(i)}_{[1,\,l-1]}$ in the data-driven setting, and present a data-driven formulation of robust tubes.
\section{DATA-DRIVEN STOCHASTIC MPC}\label{sec:DDSMPC}
Trajectory data generated by a persistently exciting input signal allow for a data-driven representation of discrete-time LTI systems, based on a powerful result of behavioral system theory~\cite{willems2005note}, which we state in input-state-space~\cite{willemsStateSpace}.
\begin{lemma}[Fundamental Lemma]\label{lem:fundamental}
Consider system~\eqref{eq:system1} with $\gls*{dk}=\bm{0}$. If the input sequence $\gls*{udata}_{[0,\,N-1]}$ is persistently exciting of order $L+n_\mathrm{x}+1$, then any $(L+1)$-long input-state sequence $\left(\bm{u}_{[k,\,k+L]},\bm{x}_{[k,\,k+L]}\right)$ is a valid trajectory of the system if and only if there exists $\bm{\alpha} \in \mathbb{R}^{N-L}$ such that
	\begin{equation} \label{eq:fundlemm_eq}
		\mat{\underline{\bm{u}}_{[k,\,k+L]}\\ \underline{\bm{x}}_{[k,\,k+L]}} = \mat{\bm{H}_{L+1}\left(\bm{u}^{\text{d}}_{[0,\,N-1]}\right) \\ \bm{H}_{L+1}\left(\bm{x}^{\text{d}}_{[0,\,N-1]}\right)} \bm{\alpha}.
	\end{equation}
\end{lemma}
Lemma \ref{lem:fundamental} describes a non-parametric system representation of a discrete-time LTI system, where varying $\bm{\alpha}$ on the right-hand-side of~\eqref{eq:fundlemm_eq} returns different $L$-step input-state trajectories on the left-hand-side. In the OCP of a data-driven predictive control scheme~\cite{ACC2015,coulson2019data,berberich2020data}, equation~\eqref{eq:fundlemm_eq} (or the equivalent input-output formulation) replaces the model and $\bm{\alpha}$ acts as the new decision variable.
\subsection{Data-driven representation of stabilized nominal and error dynamics} \label{sec:preprocess}
In order to represent system~\eqref{eq:system1} with past data, we extend Lemma~\ref{lem:fundamental} to the case of systems with additive disturbance. 
In the following, data Hankel matrices will be abbreviated as
$\bm{H}_u:=\bm{H}_{L+1}(\gls*{udata}_{[0,\,N-1]})$, $\bm{H}_d:=\bm{H}_{L+1}(\gls*{ddata}_{[0,\,N-1]})$, and $\bm{H}_x:=\bm{H}_{L+1}(\gls*{xdata}_{[0,\,N-1]})$.
\begin{lemma}[Extended Fundamental Lemma] \label{lem:extfundamental} 
	Consider data as in Assumption~\ref{assum:trajData}.
	Any $(L+1)$-long input-disturbance-state sequence $\left({\bm{u}}_{[k,\,k+L]},{\bm{d}}_{[k,\,k+L]},{\bm{x}}_{[k,\,k+L]}\right)$ is a valid trajectory of system~\eqref{eq:system1} if and only if there exists $\bm{\alpha} \in \mathbb{R}^{N-L}$ such that
	\begin{equation} \label{eq:extfundamental}
		\mat{\underline{\bm{u}}_{[k,\,k+L]}\\ \underline{\bm{d}}_{[k,\,k+L]}\\ \underline{\bm{x}}_{[k,\,k+L]}} = \mat{\bm{H}_u\\ \bm{H}_d \\ \bm{H}_x} \bm{\alpha}.
	\end{equation}
\end{lemma}
\begin{proof}
Lemma \ref{lem:extfundamental} follows directly from Lemma \ref{lem:fundamental} by considering an extended input $\mat{\gls*{udata} & \gls*{ddata}}^{\top}$
and reordering the rows of the Hankel matrices accordingly.
\end{proof}
In order to counteract an inflation of the error state due to the dynamics in $\bm{A}$, we pre-stabilize the system by introducing a stabilizing state feedback controller. Accordingly, the input $\bm{u}_k$ is decomposed into a state feedback component and a new artificial input $\bm{v}_k$, $\bm{u}_k = \bm{K}\bm{x}_k + \bm{v}_k$.
\begin{remark}\label{rem:ComputeK}
Stabilizing and LQR-optimal state feedback gains $\bm{K}\in\mathbb{R}^{n_\mathrm{u}\times \glsd*{xk}}$ for nominal LTI systems can be computed from data based on~\cite[Theorem 3]{de2019formulas} and ~\cite[Theorem 4]{de2019formulas}, respectively. To that end, we retrieve a nominal (undisturbed) input-state trajectory with Lemma~\ref{lem:extfundamental} by setting $\underline{\bm{d}}_{[k,\,k+L]}=\bm{0}$, choosing arbitrary inputs $\underline{\bm{u}}_{[k,\,k+L]}$, fixing an arbitrary initial state $\bm{x}_k$, and solving~\eqref{eq:extfundamental} for the resulting state sequence $\bm{x}_{[k+1,\,k+L]}$.
\end{remark}
In order to include the state feedback in the data-driven system representation~\eqref{eq:extfundamental}, i.e., represent the closed-loop behavior with open-loop data, we pretend the data generating input sequence was already given by $\bm{u}_k = \bm{K}\bm{x}_k + \bm{v}_k$. Then, we rearrange~\eqref{eq:extfundamental} appropriately, to change the input variable of interest from $\bm{u}_k$ to $\bm{v}_k=\bm{u}_k-\bm{K}\bm{x}_k$.
\begin{lemma}[Pre-stabilized fundammental lemma]\label{lem:prestabRepresent}
Consider open-loop data of system~\eqref{eq:system1} as in Assumption~\ref{assum:trajData} and let $\bm{K}\in\mathbb{R}^{n_\mathrm{u}\times n_\mathrm{x}}$ be chosen such that $\gls*{udata}_{[0,N-1]}-\bm{K}\gls*{xdata}_{[0,N-1]}$ is persistently exciting of order $L+n_\mathrm{x}+1$.
Any $(L+1)$-long sequence $(\underline{\bm{v}}_{[k,\,k+L]},\underline{\bm{d}}_{[k,\,k+L]},\underline{\bm{x}}_{[k,\,k+L]})$ is a vaild trajectory of $\gls*{xk1} = (\bm{A}+\bm{K}\bm{B}) \gls*{xk} + \bm{B} \bm{v}_k + \gls*{Bd} \gls*{dk}$ if and only if there exists $\bm{\alpha} \in \mathbb{R}^{N-L}$ such that
\begin{equation} \label{eq:extfundamentalK}
	\mat{\underline{\bm{v}}_{[k,\,k+L]}\\ \underline{\bm{d}}_{[k,\,k+L]}\\ \underline{\bm{x}}_{[k,\,k+L]}} = \mat{\bm{H}_u- \tilde{\bm{K}} \bm{H}_x\\ \bm{H}_d \\ \bm{H}_x} \bm{\alpha},
\end{equation}
where $\tilde{\bm{K}}\in\mathbb{R}^{n_\mathrm{u}(L+1)\times n_\mathrm{x}(L+1)}$ is a block-diagonal expansion of $\bm{K}$.
\end{lemma}
\begin{proof}
The proof follows from Lemma~\ref{eq:extfundamental}, by replacing $\bm{H}_u$ with $\bm{H}_v = \bm{H}_u- \tilde{\bm{K}} \bm{H}_x$, reflecting the change of input from $\bm{u}_k$ to $\bm{v}_k$.
\end{proof}
Since the dynamics of the underlying system are stochastic in nature, we split the system state into nominal state $\bm{z}$ and error state $\bm{e}$, as motivated in Section~\ref{sec:smpc}.
In the data-driven setting, the state decomposition~\eqref{eq:state_decomp} translates to nominal state trajectories $\bm{z}_{[k,\,k+L]}$ only influenced by inputs $\bm{u}_{[0,\,L]}$, and state error trajectories $\bm{e}_{[k,\,k+L]}$, only influenced by disturbances $\bm{d}_{[0,\,L]}$.
With either the disturbance sequence or the input sequence set to zero in~\eqref{eq:extfundamentalK}, we obtain the data-driven representation of the pre-stabilized nominal~\eqref{eq:nominal_dynamics} or error~\eqref{eq:error_dynamics} dynamics as
\begin{subequations} \label{eq:alphaz+e}
	\begin{align}	
	\mat{\underline{\bm{v}}_{[k,\,k+L]}\\\bm{0} \\ \underline{\bm{z}}_{[k,\,k+L]}} 
	&= \mat{\bm{H}_u - \tilde{\bm{K}} \bm{H}_x \\ \bm{H}_d \\ \bm{H}_x}\bm{\alpha}_z, \label{eq:alphaz} \\
	\mat{\bm{0}\\ \underline{\bm{d}}_{[k,\,k+L]} \\ \underline{\bm{e}}_{[k,\,k+L]}} 
	&= \mat{\bm{H}_u - \tilde{\bm{K}} \bm{H}_x\\ \bm{H}_d \\ \bm{H}_x}\bm{\alpha}_e, \label{eq:alphae}
	\end{align}
\end{subequations}
where $\bm{\alpha}_e\in\mathbb{R}^{N-L}$, $\bm{\alpha}_z\in\mathbb{R}^{N-L}$ 
and the true state sequence is given by
$
\underline{\bm{x}}_{[k,\,k+L]}=\underline{\bm{z}}_{[k,\,k+L]}+\underline{\bm{e}}_{[k,\,k+L]}=\bm{H}_x(\bm{\alpha}_z+\bm{\alpha}_e).
$
Equation~\eqref{eq:alphae} allows for the computation of state errors $ \underline{\bm{e}}_{[k,\,k+L]}$ from a sequence of disturbances $\underline{\bm{d}}_{[k,\,k+L]}$, which will be used to compute tightened constraints for the nominal state $\bm{z}$ in Section~\ref{sec:DDSMPC}-\ref{sec:tightcons}. In practice, this computation takes place offline, before the actual control phase.
Equation~\eqref{eq:alphaz} allows for prediction of nominal state sequences and replaces~\eqref{ocpintro_evolve} in the OCP (formulated in Section~\ref{sec:DDSMPC}-\ref{sec:controller}), to find optimal input sequences $\bm{v}_{[k,\,k+L]}$ during the control phase. 
\subsection{Data-driven Tightening of State Constraints} \label{sec:tightcons_noise}\label{sec:tightcons}
In order to guarantee chance constraint satisfaction of the states $\bm{x}_{l|k}$ in the prediction horizon, we subject the nominal states $\bm{z}_{l|k}$ to tightened constraints. We first tighten constraints probabilistically with respect to the additive disturbance, before further tightening constraints robustly with respect to the measurement noise.

The probabilistcally tightened constraint sets $\gls*{tightX}_l$ are defined by~\eqref{eq:stateprobcons_tight} based on parameters $\bm{\eta}_l$. To find $\bm{\eta}_l$, 
the chance constrained optimization problem~\eqref{eq:stateprobcons_eta} is solved using sampled state error sequences ${\bm{e}_{[0,\,L]}}$. 
Given the initial state error ${\bm{e}_0 = \bm{0}}$ and a sampled disturbance sequence $\bm{d}_{[0,\,L]}$, the state error sequence $\bm{e}_{[0,\,L]}$ follows from~\eqref{eq:alphae} as
\begin{equation}\label{eq:stateErrorPred}
    \bm{e}_{[0,\,L]} = \bm{H}_x \mat{\bm{H}_u - \tilde{\bm{K}}\bm{H}_x \\ \bm{H}_d \\ \left[\bm{H}_{x}\right]_{[1,\,\glsd*{xk}]}}^{\dagger} \mat{\bm{0}\\ \underline{\bm{d}}_{[0,\,L]} \\ \bm{0}}.
\end{equation}
With the constraints $\gls*{X}$ tightened to $\gls*{tightX}_l$ for the nominal state,~\eqref{eq:StochasticTubeImplies} holds (with confidence $\beta$) and the pre-specified chance constraints are satisfied for system~\eqref{eq:system1}.

Next, we account for inexact state measurements $\hat{\bm{x}}_k=\bm{x}_k+\gls*{noise}_k$~\eqref{eq:system2} during the control phase. Let $\hat{\bm{z}}_{[0|k,\,L|k]}$ denote the nominal state predictions initialized at the measured state $\hat{\bm{z}}_{[0|k]}=\gls*{xhk}$. In the following, we further tighten the tube sets $\gls*{tightX}_l$ robustly with respect to all possible realizations of the measurement noise $\gls*{noise}_k\in\glsd*{epsk}$. That is, we construct robust tube sets $\gls*{tightXnoise}_l$ as in~\eqref{eq:RobustTubeImplies} such that 
\begin{equation}\label{eq:robTubeNoiseimplies}
    \hat{\bm{z}}_{l|k}\in\gls*{tightXnoise}_l\implies \bm{z}_{l|k}\in\gls*{tightX}_l.
\end{equation}
In order to obtain a maximal tube, we set $\gls*{tightXnoise}_l := \gls*{tightX}_l \ominus \glsd*{epsl}$, where $\glsd*{epsl}$ encompasses all possible deviations from  predicted nominal states $\hat{\bm{z}}_{l|k}$ to noise-free nominal states $\bm{z}_{l|k}$,
\begin{align} \label{eq:nompred_tube}
	\bm{z}_{l|k} &\in \{\hat{\bm{z}}_{l|k}\} \oplus -\glsd*{epsl}\quad\forall\gls*{noise}_{[0|k,\,l|k]}\in\glsd*{epsk}.
\end{align}
The smallest possible sets $\glsd*{epsl}$ which satisfy~\eqref{eq:nompred_tube} are given by $\glsd*{epsl}=\{\hat{\bm{z}}_{l|k} - \bm{z}_{l|k} ~\left|~ \gls*{noise}_k\in\glsd*{epsk}\right.\}$, and are bounded, since the measurement noise $\gls*{noise}_k\in\glsd*{epsk}$ is bounded.  
As the deviation of the two nominal trajectories evolves with $\bm{A}$ via
\begin{equation}\label{eq:nompred_model}
    \hat{\bm{z}}_{l|k} - \bm{z}_{l|k} = \bm{A}^l \gls*{noise}_k, \quad\forall l \in \mathbb{N}_0^L,
\end{equation}
the difference of the two input-disturbance-state trajectories $\left(\ul{\bm{u}}_{[k,k+L]}, \bm{0},\ul{\hat{\bm{z}}}_{[0|k,L|k]}\right)$ and $\left(\ul{\bm{u}}_{[k,k+L]}, \bm{0},\ul{\bm{z}}_{[0|k,L|k]}\right)$ is again a valid input-disturbance-state trajectory of the same underlying system. Consequently, the data-driven representation in Lemma~\ref{lem:extfundamental} applies, yielding
\begin{equation} \label{eq:alpha_eps}
    \mat{\bm{0}\\ \bm{0}\\ \ul{\hat{\bm{z}}}_{[0|k, L|k]} - \ul{\bm{z}}_{[0|k, L|k]}} = \mat{\bm{H}_u\\ \bm{H}_d \\ \bm{H}_x} \bm{\alpha}_{\mu},
\end{equation}
with $\bm{\alpha}_{\mu} \in \mathbb{R}^{N-L}$, and where the noise $\gls*{epsk} = \hat{\bm{z}}_{0|k} - \bm{z}_{0|k}$ appears as the first $n_{\mathrm{x}}$ entries of $\ul{\hat{\bm{z}}}_{[0|k, L|k]} - \ul{\bm{z}}_{[0|k, L|k]}$, i.e.,
\begin{equation}
    \gls*{noise}_k=\left[\bm{H}_x\right]_{[1,n_{\mathrm{x}}]}\bm{\alpha}_{\mu}.
\end{equation}
We now have all the ingredients to construct the sets $\glsd*{epsl}$ from data. Recall, that $\gls*{epsk}\in\glsd*{epsk}=\left\{\gls*{noise} \in \glsd*{xk} ~\left|~ \bm{G}_{\mu} \gls*{noise} \le \bm{g}_{\mu} \right.\right\}$ by Assumption~\ref{assum:DisturbanceBounds}.
\begin{lemma}\label{lem:noise_tight}
Consider system~\eqref{eq:system} and data persistently exciting of order $L+n_\mathrm{x}+1$ as in Assumption~\ref{assum:trajData}.
    Let $\mathbb{A}_{\mu}=\{ \bm{\alpha} \in \mathbb{R}^{N-L}
	    \mid
	    ~\bm{G}_{\mu}\left[\bm{H}_x\right]_{[1,n_{\mathrm{x}}]} \bm{\alpha} \le \bm{g}_{\mu}\}$.
    For each $l \in \mathbb{N}_0^L$, if
    \begin{equation}\label{eq:MeasNoiseTube}
	    \glsd*{epsl} = \left[\bm{H}_x\right]_{[ln_\mathrm{x}+1,ln_\mathrm{x}+n_\mathrm{x}]}\left(\mathbb{A}_{\mu}\cap \text{ker}\mat{\bm{H}_u\\ \bm{H}_d}\,\right),
    \end{equation}
    then $\bm{z}_{l|k} \in \{\hat{\bm{z}}_{l|k}\} \oplus -\glsd*{epsl}$ and the sets $\glsd*{epsl}$ define the smallest possible tube such that \eqref{eq:robTubeNoiseimplies} holds.
\end{lemma}
\begin{proof}
    We have to prove that $\hat{\bm{z}}_{l|k} - \bm{z}_{l|k}\in \glsd*{epsl}$ for all $l \in \mathbb{N}_1^L$ and all $\gls*{noise}_k\in\glsd*{epsk}$.
    Denote with $\bm{H}_{x,l}=\left[\bm{H}_x\right]_{[ln_\mathrm{x}+1,ln_\mathrm{x}+n_\mathrm{x}]}$ the rows of $\bm{H}_x$ in~\eqref{eq:alpha_eps} corresponding to $\hat{\bm{z}}_{l|k} - \bm{z}_{l|k}$. First, assume $\gls*{noise}_k$ fixed and note that by Lemma~\ref{lem:extfundamental} and \eqref{eq:alpha_eps}, $\hat{\bm{z}}_{l|k} - \bm{z}_{l|k}=\bm{H}_{x,l}\bm{\alpha}_{\mu}=\bm{A}^l \gls*{noise}_k$ if $\bm{H}_{x,0} \bm{\alpha} = \gls*{noise}_k$ and $\mat{\bm{H}_u \\\bm{H}_d} \bm{\alpha} = \bm{0}$.
    Thus, as $\gls*{noise}_k\in\glsd*{epsk}$ per assumption, $\bm{A}^l\glsd*{epsk} = {\bm{H}_{x,l}\,\left\{ \bm{\alpha} \in \mathbb{R}^{N-L}
	\,\left|\,
	\mat{\bm{H}_u \\\bm{H}_d} \bm{\alpha} = \bm{0},
	\,\bm{H}_{x,0} \bm{\alpha}\in\glsd*{epsk}\right.\right\}}=\glsd*{epsl}$.
	If any element of $\glsd*{epsl}$ is omitted, $\glsd*{epsl}\subsetneq \bm{A}^l\glsd*{epsk}$ and there exists a sequence of measurement noise realizations such that \eqref{eq:robTubeNoiseimplies} is violated. As a consequence, the tube sets $\glsd*{epsl}$ are minimal.
\end{proof}
The nominal state within the OCP is thus constrained as
\begin{equation} \label{eq:tightcons_noise}
	\hat{\bm{z}}_{l|k} \in \gls*{tightXnoise}_l := \gls*{tightX}_l \ominus \glsd*{epsl} ~~\forall l \in \mathbb{N}_1^L.
\end{equation}
\begin{remark}\label{rem:ConvHull}
Note that with $l \in \mathbb{N}_1^L$ we omit constraints on the initial state $\hat{\bm{z}}_{0|k}=\gls*{xhk}$ as it is not affected by future inputs, but fixed. 
For recursive feasibility and stability proofs or for the construction of a feasible candidate solution, it is often desired that the tightened nominal state constraints $\gls*{tightXnoise}$ contract, such that ${\gls*{tightXnoise}_{l+1} \subseteq \gls*{tightXnoise}_{l} ~\forall l \in \mathbb{N}_1^{L-1}}$. This can be guaranteed by substituting $\glsd*{epsl}$ in~\eqref{eq:tightcons_noise} with the convex hull of $\cup_{i=1}^{l} \mathcal{E}_{\mu,i}$, albeit the tube may become more conservative.
\end{remark}
\subsection{Tightened Input Constraints}
Both the additive disturbance and the measurement noise influence the measured state, which introduces uncertainty into actually applied inputs via the state feedback component.
As $\bm{z}_{0|k}=\gls*{xhk}=\bm{A}\bm{x}_{k-1}+\bm{B}\bm{u}_{k-1}+\gls*{Bd}\bm{d}_{k-1}+\gls*{epsk}$, the additive disturbance $\bm{d}_{k}$ does not influence the actually applied input $\bm{u}^*_{0|k}$, but only succeeding inputs $\bm{u}^*_{1|k},\ldots,\bm{u}^*_{L|k}$. Thus, a probabilistic constraint tightening may be employed to reduce conservatism.
Analogously to~\eqref{eq:stateprobcons_tight} and~\eqref{eq:stateprobcons_eta}, we probabilistically tighten input constraints for all $l\in \mathbb{N}_0^{L}$ as
\begin{equation} \label{eq:inputprobcons_tight}
    \mathcal{U}_l = \left\{\bm{u} \in \glsd*{uk}  ~\left|\right.~ \bm{G}_u \bm{u} \le \bm{\sigma}_l\right\},
\end{equation}
with each $\bm{\sigma}_l$ chosen by chance-constrained optimization
\begin{subequations} \label{eq:inputprobcons_mu}
	\begin{align}
	\bm{\sigma}_l &= \max\limits_{\tilde{\bm{\sigma}}} \tilde{\bm{\sigma}} \\
	&\text{s.t. } \Pr{\tilde{\bm{\sigma}} \le \bm{g}_{u} - \bm{G}_{u} \bm{K} \bm{e}_{l} } \ge p.
	\end{align}
\end{subequations}
The measurement noise $\gls*{epsk}$ has a direct influence however, so that we employ a robust input constraint tightening with respect to all possible realizations of $\gls*{epsk}$,
\begin{equation}\label{eq:tightU}
    \gls*{tightU} = \mathcal{U}_{l} \ominus \bm{K}\glsd*{epsl} ~~\forall l \in \mathbb{N}_0^L.
\end{equation}
\begin{remark} The input constraint sets $\mathcal{U}_l$ are such that ${\mathcal{U}_{l+1} \subseteq \mathcal{U}_{l} ~\forall l \in \mathbb{N}_0^{L-2}}$. Note that $\bm{e}_{0|k}=\bm{0}$ for the initial state and the actual applied input always satisfies the specified constraints $\bm{u}_k\in\gls*{U}$.
\end{remark}
\subsection{Data-driven SMPC} \label{sec:controller}
Given the evolution of the nominal state~\eqref{eq:alphaz}, tightened state and input constraint sets~\eqref{eq:tightcons_noise},\eqref{eq:tightU}, and the explicit expectation of the cost function, we obtain a tractable deterministic data-driven reformulation of the stochastic OCP~\eqref{eq:ocp_intro}
\begin{subequations} \label{eq:ocp}
	\begin{align}
		\underset{\hat{\bm{Z}}_k, \bm{V}_k, \bm{\alpha}_k}{\text{minimize}}~&\sum\limits_{l = 0}^{L-1} 
		J_{\text{s}}\left(\hat{\bm{z}}_{l|k}, \bm{v}_{l|k}+\bm{K}\hat{\bm{z}}_{l|k}\right)+
		J_{\text{f}}\left(\hat{\bm{z}}_{L|k}\right)\label{eq:ocp_cost}\\	
		\text{s.t.}~ &\hat{\bm{z}}_{0|k} = \gls*{xhk},\label{eq:ocp_init} \\	
		& \mat{\ul{\bm{v}}_{[0|k,\,L|k]}\\ \bm{0} \\
			\ul{\hat{\bm{z}}}_{[0|k,\,L|k]}} = \mat{\bm{H}_u - \tilde{\bm{K}}\bm{H}_x \\ \bm{H}_d \\ \bm{H}_{x}} \bm{\alpha}_k, \label{eq:ocp_hankel}\\
		& \hat{\bm{z}}_{l|k} \in \gls*{tightXnoise}_l  ~\hspace*{1.65cm}\forall l \in \mathbb{N}_1^L, \label{eq:ocp_stateCon}\\
		& \bm{v}_{l|k} + \bm{K} \hat{\bm{z}}_{l|k} \in \gls*{tightU} ~~~\forall l \in \mathbb{N}_0^{L-1},\label{eq:ocp_inputCon} \\
		& \hat{\bm{z}}_{L|k} \in \gls*{termSetNoise},\label{eq:ocp_terminalCon}\\
	    & \hat{\bm{z}}_{1|k} \in \gls*{firstSet},\label{eq:ocp_firstCon}
	\end{align}
\end{subequations}
with prediction horizon $L$, predicted nominal state sequence ${\hat{\bm{Z}}_k := \hat{\bm{z}}_{[0|k,\,L|k]}}$, input sequence $\bm{V}_k := \bm{v}_{[0|k,\,L-1|k]}$, stage cost
\begin{equation} \label{eq:stagecost}
	J_{\text{s}}\left(\hat{\bm{z}}_{l|k}, \bm{u}_{l|k}\right) = \normsq{\hat{\bm{z}}_{l|k}}_{\bm{Q}} + \normsq{\bm{u}_{l|k}}_{\bm{R}},
\end{equation}
with positive definite weighting matrices $\bm{Q} \in \mathbb{R}^{\glsd*{xk} \times \glsd*{xk}}$, $\bm{R} \in \mathbb{R}^{n_\mathrm{u} \times n_\mathrm{u}}$, and a terminal cost function
\begin{equation} \label{eq:terminalcost}
	J_{\text{f}}\left(\hat{\bm{z}}_{L|k}\right) = \normsq{\hat{\bm{z}}_{L|k}}_{\bm{P}},
\end{equation}
with positive definite $\bm{P} \in \mathbb{R}^{n_\mathrm{x} \times n_\mathrm{x}}$.
Equations \eqref{eq:ocp_terminalCon}, \eqref{eq:ocp_firstCon} are terminal and first-step constraints, which play an important role for recursive feasibility and stability of the proposed receding horizon control scheme. They are included here for the sake of completeness and will be introduced in more detail in Sec.~\ref{sec:properties}.

The OCP~\eqref{eq:ocp} has polyhedral constraints and quadratic costs, and can thus be written as a convex quadratic program parameterized by the measured state $\gls*{xhk}$ and with decision variable $\bm{\alpha}_k$.
As a consequence, and since $\bm{R}$ is positive definite, it admits a unique solution $\bm{\alpha}^*_k$ (and thereby $\bm{V}^*_k$) that depends entirely on~$\gls*{xhk}$ (see for example~\cite[Chapter 7]{RawlingsMayneDiehl2017} for details).
In the following, we denote the OCP~\eqref{eq:ocp} as $\gls*{OCP}(\gls*{xhk})$, and the resulting implicit control law by
\begin{equation}\label{eq:controllaw}
\gls*{ulaw}\left(\gls*{xhk}\right):=\bm{u}^{*}_k = \bm{K} \gls*{xhk}+\bm{v}^{*}_{0|k}.
\end{equation}
Since the OCP is a convex program, small changes in the initial state $\gls*{xhk}$ lead to small changes in the optimal decision variable $\bm{\alpha}^*_k$ and therefore the associated optimal input sequence $\bm{V}^*_k$.
\begin{proposition}\label{prop:lipschitz_ulaw}
Let $\gls*{roa}$ be the set of initial states $\gls*{xhk}\in\gls*{roa}$ for which $\gls*{OCP}(\gls*{xhk})$ is solvable, i.e., $\gls*{roa}$ is the set of feasible initial states.
For all $\gls*{xhk}\in\gls*{roa}$, the implicit control law $\gls*{ulaw}\left(\gls*{xhk}\right)$ is Lipschitz continuous.
\end{proposition}
\begin{proof}
Since the OCP itself is uncertainty free, has quadratic costs and polytopic constraints, Lipschitz continuity of the implicit OCP control law is a well known property in MPC and follows for example from~\cite[Proposition 17]{goulart2006optimization}. Lipschitz continuity of $\gls*{ulaw}\left(\gls*{xhk}\right)$ follows immediately from Lipschitz continuity of the linear term $\bm{K} \gls*{xhk}$.
\end{proof}
The complete tube-based data-driven stochastic predictive control scheme is summarized in Algorithm~\ref{alg:ddspc}. 
\begin{algorithm}[H]
\caption{Data-driven stochastic predictive control}
\begin{algorithmic}[1]\label{alg:ddspc}
\renewcommand{\algorithmicrequire}{\textbf{Offline:}}
\renewcommand{\algorithmicensure}{\textbf{Online:}}
\REQUIRE Given data as in Assumptions~\ref{assum:trajData}, \ref{assum:DisturbanceBounds} (a).
\STATE Compute LQR feedback gain $\bm{K}$ as in Remark~\ref{rem:ComputeK}.
\STATE For each $l\in\mathbb{N}_1^L$, compute state errors $\bm{e}_l^{(i)}$~\eqref{eq:stateErrorPred} from disturbance data $\bm{d}_{[0,l-1]}^{(i)}$ for all $i=1,\ldots,N_S$.
\STATE Retrieve constraint sets for nominal state $\gls*{tightX}_l$~\eqref{eq:stateprobcons_tight} and input $\mathcal{U}_l$~\eqref{eq:inputprobcons_tight} by solving chance-constrained optimization problems~\eqref{eq:stateprobcons_eta},~\eqref{eq:inputprobcons_mu}.
 \STATE Robustly tighten $\gls*{tightX}_l$, $\mathcal{U}_l$ with respect to the bounded measurement noise, obtaining $\gls*{tightXnoise}_l$~\eqref{eq:tightcons_noise} and $\gls*{tightU}$ \eqref{eq:tightU}
  \STATE Compute $\bm{P}$~\eqref{eq:computeP} and define terminal costs~\eqref{eq:terminalcost}.
 \STATE Compute and tighten the terminal set~\eqref{eq:termset}, \eqref{eq:tightcons_terminal_noise}.
 \STATE Compute first step constraint set~\eqref{eq:firststepcons}.
 \STATE Based on all the above, construct the OCP $\gls*{OCP}(\cdot)$~\eqref{eq:ocp}.
 \ENSURE For all time steps $k$:
 \STATE Obtain noisy state measurement $\gls*{xhk}$.
 \STATE Solve the OCP $\gls*{OCP}(\gls*{xhk})$~\eqref{eq:ocp}.
 \STATE Apply input $\bm{u}_k=\gls*{ulaw}\left(\gls*{xhk}\right)$~\eqref{eq:controllaw} to the system.
 \STATE Set $k\leftarrow k+1$ and go back to Step 9.
\end{algorithmic}
\end{algorithm}
\section{RECURSIVELY FEASIBLE AND CLOSED-LOOP STABLE DESIGN}\label{sec:properties}
In this section, we design terminal costs, terminal constraints and first-step constraints such that the OCP remains feasible, and the resulting closed-loop system is stable.
Stability, discussed in Section~\ref{sec:properties}-\ref{sec:iss}, follows similarly to classical arguments in MPC, by choosing the terminal constraint set as a robust positive
invariant set under the local control law and the terminal cost
function as a local control Lyapunov function in that terminal set.
Recursive feasibility, discussed in Section~\ref{sec:properties}-\ref{sec:recfeas}, follows from an additional first step constraint that guarantees robust positive invariance of the feasible set.

In the following, we will elaborate on the terminal constraint set, which plays a role for both feasibility and stability. 
Inspired by~\cite[Prop.~2]{lorenzen2016constraint} we first define a feasible set $\gls*{X}_{\text{f}}$ under the control law $\gls*{uk} = \bm{K} \gls*{xhk}$. 
To that end, consider the set $\tilde{\gls*{X}}_{\text{f}}$ of all states $\hat{\bm{x}}$ for which $\bm{K}\hat{\bm{x}}$ satisfies tightened input constraints, and the nominal successor state $\hat{\bm{x}}^{+}=(\bm{A}+\bm{B}\bm{K})\hat{\bm{x}}$ is inside the tube~\eqref{eq:tightcons_noise}.
\begin{equation} \label{eq:termset}
	\tilde{\gls*{X}}_{\text{f}} = \left\{ \hat{\bm{x}} \in \glsd*{xk} \,\left|\, \begin{array}{l} 
	\left(\exists\bm{\alpha}_{\text{f}} \in \mathbb{R}^{N-L}\right) \\
	\mat{\bm{0} \\ \bm{0} \\ \hat{\bm{x}} \\ \hat{\bm{x}}^{+} } = \mat{\left[\bm{H}_u - \tilde{\bm{K}} \bm{H}_x \right]_{[1,n_{\mathrm{u}}]}\\ \left[\bm{H}_d\right]_{[1,n_{\mathrm{d}}]} \\ \left[\bm{H}_{x}\right]_{[1,n_{\mathrm{x}}]} \\ \left[\bm{H}_{x}\right]_{[n_{\mathrm{x}}+1,2n_{\mathrm{x}}]}} \bm{\alpha}_{\text{f}}, \\ \hat{\bm{x}}^{+} \in \gls*{tightXnoise}_1,\, \bm{K}\hat{\bm{x}} \in \hat{\mathcal{U}}_0
	\end{array} \right.
	\right\}.
\end{equation}
\begin{proposition}[RPI feasible set under local control]
	Let $\gls*{X}_{\text{f}}\subseteq\tilde{\gls*{X}}_{\text{f}}$
	be a robust positive invariant polytope for system~\eqref{eq:system} controlled via $\gls*{uk} = \bm{K} \gls*{xhk}$, i.e., ${\gls*{xk1}=\bm{A}\gls*{xk}+\bm{B}\bm{K}(\gls*{xk}+\gls*{epsk})+\gls*{dk}}$.
	Then, for any initial state in $\gls*{X}_{\text{f}}$, the original state constraints~\eqref{eq:ocp_intro_chancecon} and input constraints~\eqref{eq:ocp_intro_InputCon} are satisfied for all ${k>0}$.
\end{proposition}
\begin{proof}
	By definition of $\tilde{\gls*{X}}_{\text{f}}$, a measured state $\hat{\bm{x}} \in \gls*{X}_{\text{f}}$ implies $\hat{\bm{x}}^{+} \in \gls*{tightXnoise}_1$, which in turn implies $\bm{x}^{+} \in \gls*{tightX}_1$ for the actual (not measured) state by definition of the tube set~\eqref{eq:tightcons_noise}. 
	By construction of $\gls*{tightX}_1$ in \eqref{eq:stateprobcons_tight}, $\bm{x}^{+} \in \gls*{tightX}_1$ implies satisfaction of the state constraints $\Pr{\bm{x}^{+} \in \gls*{X}} \ge p$. By robust positive invariance of $\gls*{X}_{\text{f}}$, $\hat{\bm{x}}^{+}\in\gls*{X}_{\text{f}}$ and the arguments holds for all future time steps by induction.
	Similarly, $\hat{\bm{x}} \in \gls*{X}_{\text{f}}$ implies $\bm{K}\hat{\bm{x}} \in \hat{\mathcal{U}}_0$ so that input constraints $\bm{u}=\bm{K}\hat{\bm{x}}\in\gls*{U}$ are satisfied since $\hat{\mathcal{U}}_0\subseteq\gls*{U}$. 
	Again, robust positive invariance of $\gls*{X}_{\text{f}}$ guarantees that $\hat{\bm{x}}^{+}\in\gls*{X}_{\text{f}}$ and the constraints hold for all time.
\end{proof}
We use $\gls*{X}_{\text{f}}$ to specify the terminal constraint set $\gls*{termSetNoise}$ of the OCP such that
\begin{equation}
    \hat{\bm{z}}_{L|k}\in\gls*{termSetNoise}\implies \Pr{\bm{x}_{L|k}\in\gls*{X}_{\text{f}}} \ge p
\end{equation}
for all realizations of the disturbance and measurement noise.
$\gls*{termSetNoise}$ is computed by first probabilistically tightening $\gls*{X}_{\text{f}}$  with respect to the additive disturbance as in~\eqref{eq:stateprobcons_tight}, and then robustly tightening with respect to the measurement noise.
Let $\gls*{X}_{\text{f}} = \left\{\bm{x} \in \glsd*{xk} ~\left|~ \bm{G}_{x,\text{f}} \bm{x} \le \bm{g}_{x,\text{f}} \right.\right\}$, then
\begin{equation} \label{eq:stateprobcons_tight_terminal}
    \gls*{termSet} = \left\{\bm{z} \in \glsd*{xk}  ~\left|\right.~ \bm{G}_{x,\text{f}} \bm{z} \le \bm{\eta}_{\text{f}}\right\},
\end{equation}
where $\bm{\eta}_{\text{f}}$ solves the chance constrained optimization
\begin{subequations} \label{eq:stateprobcons_eta_terminal}
	\begin{align}
 	\bm{\eta}_{\text{f}} &= \max\limits_{\tilde{\bm{\eta}}} \tilde{\bm{\eta}} \\
	\text{s.t. } & \Pr{\tilde{\bm{\eta}} \le \bm{g}_{x,\text{f}} - \bm{G}_{x,\text{f}} \bm{e}_{L} } \ge p,
	\end{align}
\end{subequations}
and the robustly tightened terminal set is then defined as
\begin{equation} \label{eq:tightcons_terminal_noise}
	\gls*{termSetNoise} := \gls*{termSet} \ominus \glsd*{epsL},
\end{equation}
where $\glsd*{epsL}$ encompasses all uncertainty induced by the measurement noise and is given by \eqref{eq:MeasNoiseTube}.
\subsection{Recursive Feasibility} \label{sec:recfeas}
We want to guarantee that if the OCP $\gls*{OCP}(\gls*{xhk})$ is feasible at time step $k$ (there exists an admissible input and corresponding state sequence), then the OCP $\gls*{OCP}(\gls*{xhk1})$ is feasible at the next time step $k+1$.
In nominal MPC, recursive feasibility (Def.~\ref{def:recFeasPrelim}) can be guaranteed with a control invariant terminal constraint set. 
The usual argument, as for example presented in~\cite{grune2017nonlinear}, is based on a set $\mathcal{C}_l$ that denotes all states from which the terminal set can be reached in $l$ steps without violation of input or state constraints. That is, $\mathcal{C}_l$ is the set of feasible initial states of the OCP with horizon $l$, and $\mathcal{C}_0$ is equal to the terminal set. 
Since the terminal set is control invariant, $\bm{x}\in\mathcal{C}_l$ implies $\bm{x}\in\mathcal{C}_{l+1}$, and therefore $\mathcal{C}_{l}\subseteq \mathcal{C}_{l+1}$.
If $\bm{x}\in\mathcal{C}_l$ and we apply the first input of any admissible input sequence, the successor state $\bm{x}^+$ is in $\mathcal{C}_{l-1}\subseteq \mathcal{C}_l$. Thus $\mathcal{C}_l$ is positive invariant under the MPC law (yielding admissible inputs), and for all $\bm{x}_0\in\mathcal{C}_l$ the OCP remains feasible for all time.

In the case of disturbed (as opposed to nominal) systems, it is not guaranteed that $\bm{x}\in\mathcal{C}_l\implies \bm{x}^+\in\mathcal{C}_{l-1}$ , since $\bm{x}^+$ is not equal to the next predicted state. 
For probabilstically tightened constraint sets, this implication is nontrivial to restore, since 
the probability of constraint violation $l$ steps ahead of time step $k$ differs from the probability of constraint violation $l-1$ steps ahead of time step $k+1$~\cite{kouvaritakis2010explicit}.
Furthermore, in our setting, the inital state of the OCP is not $\bm{x}$, but the measured state $\hat{\bm{x}}$, so that $\gls*{xhk}\in\mathcal{C}_l\implies \gls*{xhk1}\in\mathcal{C}_{l-1}$ depends on both $\gls*{epsk}$ and $\gls*{epsk1}$.
A remedy presented in~\cite{lorenzen2016constraint} in the model-based setting for the case without measurement noise, is to directly ensure robust positive control invariance of $\mathcal{C}_L$, or equivalently robust positive invariance of $\mathcal{C}_L$ for the closed loop system under MPC law, by introducing an additional constraint on the first step of the prediction horizon. 
In the following, we construct such a first step constraint from data and include robustness with respect to measurement noise.
Since input and state sequences within the OCP~\eqref{eq:ocp} are uniquely determined by the decision vector $\bm{\alpha}_k$ in~\eqref{eq:ocp_hankel}, we frame recursive feasibility as the problem of ensuring the existence of a feasible $\bm{\alpha}_{k+1}$ at the next time step.
\begin{proposition}[Recursive feasibility]\label{prop:RecFeas}
Let the set $\gls*{feasA}(\gls*{xhk})$ denote all feasible decision variables $\bm{\alpha}_k$ for $\gls*{OCP}(\gls*{xhk})$, i.e., $\gls*{feasA}(\gls*{xhk}) = \{\bm{\alpha}_k\in\mathbb{R}^{N-L} \mid \text{
     \eqref{eq:ocp_init}-\eqref{eq:ocp_firstCon} are satisfied}\}$.
	The OCP~\eqref{eq:ocp} is recursively feasible according to Definition~\ref{def:recFeasPrelim} for system~\eqref{eq:system} under the proposed control law~\eqref{eq:controllaw} if
	\begin{equation}\label{eq:feasA}
		\gls*{feasA}(\gls*{xhk}) \neq \emptyset ~\implies~ \gls*{feasA}(\gls*{xhk1}) \neq \emptyset
	\end{equation}
	for every realizations of the disturbance ${\gls*{dk} \in \glsd*{dk}}$ and measurement noises $\gls*{epsk}, \gls*{epsk1} \in \glsd*{epsk}$.
\end{proposition}
 Both $\gls*{epsk}$ and $\gls*{epsk1}$ influence $\gls*{feasA}(\gls*{xhk1})$, since $\gls*{epsk}$ influences the control input $\gls*{ulaw}(\gls*{xhk})=\gls*{ulaw}(\gls*{xk}+\gls*{epsk})$, which in turn influences the successor state $\gls*{xk1}$, and thereby the initial state $\gls*{xhk1}=\gls*{xk1}+\gls*{epsk1}$ of the next OCP. 
In the following, we compute the set of feasible initial states
\begin{equation} \label{eq:reachset}
	\mathcal{C}_L = \left\{ \hat{\bm{z}}_{0|k} \in \glsd*{xk} \mid
	(\exists\,\bm{\alpha}_z \in \mathbb{R}^{N-L})\, \text{\eqref{eq:ocp_init}-\eqref{eq:ocp_terminalCon}}
	\right\}.
\end{equation}
In the model-based case, $\mathcal{C}_L$ can be computed via backward recursion~\cite{gutman1987algorithm} based on the nominal dynamics of the system and set algebra. In the data-driven case, equivalently, $\mathcal{C}_L$ may be computed by representing the system matrices $(\bm{A},\bm{B},\gls*{Bd})$ with data, based on a suitable extension of~\cite[Theorem 1]{de2019formulas} to include disturbance data.
In order to avoid this implicit system identification step, we present a projection-based technique that directly uses the data-driven system representation~\eqref{eq:alphaz}.
To that end, we first merge all state constraint sets $\gls*{tightXnoise}_l$, $l \in \mathbb{N}_0^L$, and the terminal constraint $\gls*{termSetNoise}$ into a single constraint set $\gls*{tightX}_{[1,\,L]} = \{ \ul{\bm{z}}_{[1|k,\, L|k]} \in \mathbb{R}^{n_\mathrm{x} L}  \mid \tilde{\bm{G}}_x \ul{\bm{z}}_{[1|k,\, L|k]} \le \tilde{\bm{\eta}}_x \}$ for the stacked vector of the predicted state sequence. The construction of $\tilde{\bm{G}}_x$ and $\tilde{\bm{\eta}}_x$ is straightforward, since all constraint sets are polytopic. 
Likewise, 
we define the constraint set for the vector of stacked inputs $\ul{\bm{u}}_{[0|k, L|k]}$ as $\hat{\mathcal{U}}_{[0,\,L]} = \{ \ul{\bm{u}}_{[0|k,\, L|k]} \in \mathbb{R}^{n_\mathrm{u}(L+1)}  \mid \tilde{\bm{G}}_u \ul{\bm{u}}_{[0|k,\, L|k]} \le \tilde{\bm{\sigma}}_u \}$.
\begin{multline} \label{eq:alphaset_z}
\mathcal{A}_z = \{ \bm{\alpha} \in \mathbb{R}^{N-L} \mid \tilde{\bm{G}}_u \bm{H}_u \bm{\alpha} \le \tilde{\bm{\sigma}}_u,\,
	\bm{H}_d \bm{\alpha} = \bm{0},\\ \tilde{\bm{G}}_x\left[\bm{H}_{x}\right]_{[n_\mathrm{x}+1,n_\mathrm{x}(L+1)]} \bm{\alpha} \le \tilde{\bm{\eta}}_x
\}
\end{multline}
and the $L$-step feasible initial state set $\mathcal{C}_L$~\eqref{eq:reachset} is computed by projecting $\mathcal{A}_z$ onto the coordinates of the initial state,
\begin{equation}
    \mathcal{C}_L = \left[\bm{H}_{z}\right]_{[1,n_\mathrm{x}]} \mathcal{A}_z.
\end{equation}
As mentioned above, $\mathcal{C}_L$ is not (necessarily) robust positive control invariant with respect to all uncertainties (disturbances and measurement noise).
In particular, since recursive feasibility depends on the measured state, see~\eqref{eq:feasA}, we require robust positive control invariance of (a subset of) $\mathcal{C}_L$ with respect to the evolution of the measured state
\begin{align}\label{eq:xmeasure_dynamics}
	\hat{\bm{x}}_{k+1} &= \bm{A} \left(\hat{\bm{x}}_k - \gls*{epsk}\right)+ \bm{B} \bm{u}_k + \gls*{Bd} \bm{d}_k + \gls*{noise}_{k+1}.
\end{align}
That is, the control invariant set needs to be robust with respect to the 
extended disturbance $\gls*{wk}:= \gls*{Bd} \gls*{dk} - \bm{A} \gls*{epsk} + \gls*{epsk1}$, $\gls*{wk}\in\glsd*{wk} = \gls*{Bd} \glsd*{dk} \oplus (-\bm{A}\glsd*{epsk}) \oplus \glsd*{epsk}$.
The disturbance support set $\glsd*{wk}$ can be computed from data as
\begin{equation} \label{eq:ext_dist_set}
	\glsd*{wk} = \mathcal{E}_{d,1} \oplus (-\glsd*{epsl1}) \oplus \glsd*{epsk},
\end{equation}
where $\glsd*{epsl1}$ is as in~\eqref{eq:MeasNoiseTube} and $\mathcal{E}_{d,1}$, equal to $\gls*{Bd} \glsd*{dk}$, is obtained by 
defining the appropriate set of decision variables
\begin{multline} \label{eq:alphaset_e}
	\mathcal{A}_e = \{ \bm{\alpha} \in \mathbb{R}^{N-L} \mid 
	\left[\bm{H}_u\right]_{[1,n_\mathrm{u}]} \bm{\alpha} = \bm{0}, \\
	\bm{G}_d\left[\bm{H}_d\right]_{[1,n_\mathrm{o}]} \bm{\alpha} \le \bm{g}_d,\, \left[\bm{H}_{x}\right]_{[1,n_\mathrm{x}]} \bm{\alpha} = \bm{0}\},
\end{multline}
with zero input and initial state, and then projecting onto the first step state errors $\bm{e}_{1|k}$ via
\begin{equation} \label{eq:stateerror_set}
    \mathcal{E}_{d,1} = \left[\bm{H}_x\right]_{[n_\mathrm{x}+1,2n_\mathrm{x}]}\mathcal{A}_{e}.
\end{equation}
Based on $\glsd*{wk}$~\eqref{eq:ext_dist_set}, we can construct a robust control invariant subset of $\mathcal{C}_L$ as
\begin{equation}\label{eq:robustCIreachableset}
\mathcal{C}^{\infty}_L = \cap_{i=0}^{\infty} \mathcal{C}^{i}_L,
\end{equation}
where $\mathcal{C}^{0}_L = \mathcal{C}_L$ and
\begin{equation} \label{eq:robustreachableset}
	\mathcal{C}^{i+1}_L = \left\{ \hat{\bm{x}} \in \mathcal{C}^{i}_L \,\left|\, \begin{array}{l} 
	\exists ~ \bm{\alpha} \in \mathbb{R}^{N-L}, ~\bm{u} \in \hat{\mathcal{U}}_0: \\
	\mat{\bm{u} \\ \bm{0} \\ \hat{\bm{x}} \\ \hat{\bm{x}}^{+} } = \mat{\left[\bm{H}_u \right]_{[1,n_\mathrm{u}]}\\ \left[\bm{H}_d\right]_{[1,n_\mathrm{o}]} \\ \left[\bm{H}_{x}\right]_{[1,n_\mathrm{x}]} \\ \left[\bm{H}_{x}\right]_{[n_\mathrm{x}+1,2n_\mathrm{x}]}} \bm{\alpha} \vspace*{1mm}\\  \hat{\bm{x}} \in \mathcal{C}^i_L, ~ \hat{\bm{x}}^{+} \in \mathcal{C}^{i}_L \ominus \glsd*{wk}
	\end{array} \right.
	\right\}.	
\end{equation}
In practice, $\mathcal{C}^{\infty}_L$~\eqref{eq:robustCIreachableset} is computed recursively via $\mathcal{C}^{i}_L$~\eqref{eq:robustreachableset}, until $\mathcal{C}^{i}_L = \mathcal{C}^{i+1}_L$ for some $i \in \mathbb{N}$, which implies that $\mathcal{C}^{\infty}_L = \mathcal{C}^{i}_L$ is robust positive control invariant. The idea of~\eqref{eq:robustreachableset} is to tighten the set $\mathcal{C}_L$ until all contained states admit an input for which the successor state is in $\mathcal{C}_L\ominus \glsd*{wk}$, i.e., the successor state is far enough away from the boundary of the set to render it robustly positive invariant with respect to all uncertainty in $\glsd*{wk}$.

The constraint on the first predicted state $\hat{\bm{z}}_{1|k}$,
\begin{equation} \label{eq:firststepcons}
    \hat{\bm{z}}_{1|k} \in  \gls*{firstSet}=\mathcal{C}_L^{\infty} \ominus \glsd*{wk},
\end{equation}
guarantees that for all possible realizations of the disturbance $\bm{d}_k$ and measurement noises $\gls*{noise}_k, \gls*{noise}_{k+1}$, the next measured state $\hat{\bm{x}}_{k+1}$ is inside the robust positive invariant subset $\mathcal{C}_L^{\infty}$ of the $L$-step reachable initial state set $\mathcal{C}_L$.
\begin{theorem}[Recursive feasibility]\label{prop:recfeas}
Let the first step constraint set of the receding horizon OCP~\eqref{eq:ocp} be given by~\eqref{eq:firststepcons}.
The receding horizon OCP is recursively feasible, and if
$\bm{x}_0\in \mathcal{C}_L^{\infty}\ominus \glsd*{epsk}$, the receding horizon OCP is feasible for all time.
\end{theorem}
\begin{proof}
If $\bm{x}_0\in \mathcal{C}_L^{\infty}\ominus \glsd*{epsk}$, then $\hat{\bm{x}}_0\in\mathcal{C}_L^{\infty}\subseteq \mathcal{C}_L$ and the OCP $\gls*{OCP}(\hat{\bm{x}}_0)$ is feasible by construction of $\mathcal{C}_L=\left\{ \gls*{xhk} \in \glsd*{xk} \mid \gls*{feasA}(\gls*{xhk}) \neq \emptyset\right\}$.
For recursive feasibility, assume that the OCP is feasible at time step $k$ with optimizer $\bm{\alpha}^*_k$ such that $\bm{\alpha}^*_k \in \gls*{feasA}(\gls*{xhk})$.
By construction of the first-step constraint ${\hat{\bm{z}}_{1|k}^*=\left[\bm{H}_x\right]_{[n_\mathrm{x}+1,2n_\mathrm{x}]}\bm{\alpha}^*_k \in \mathcal{C}_L^{\infty} \ominus \glsd*{wk}}$, the next measured state satisfies ${\gls*{xhk1}=\hat{\bm{z}}_{1|L} + \gls*{Bd} \gls*{dk} - \bm{A}\gls*{epsk} + \gls*{epsk1} \in \mathcal{C}_L^{\infty}}$ for all possible realizations of $\gls*{epsk}$, $\gls*{epsk1}$, $\gls*{dk}$. Since $\mathcal{C}_L^{\infty}\subseteq \mathcal{C}_L $, ${\gls*{feasA}(\gls*{xhk1}) \neq \emptyset}$ and the OCP is recursively feasible by Proposition~\ref{prop:RecFeas}. 
\end{proof}
\subsection{Input-to-State Stability} \label{sec:iss}
In order to establish ISS according to Definition~\eqref{eq:iss_cond}, we use a special ISS-Lyapunov function as proposed in~\cite{jiang2001input,goulart2006optimization}.
\begin{definition}[ISS-Lyapunov function]
    A function $V: \gls*{roa} \rightarrow \mathbb{R}_0^+$ is an \emph{ISS-Lyapunov function} for system $\gls*{xk1} = \bm{f}\left(\gls*{xk}, \gls*{wk}\right)$~\eqref{eq:sys_iss} if there exist functions $\alpha_1,\alpha_2,\alpha_3\in\mathscr{K}_{\infty}$ and $\gamma\in\mathscr{K}$ such that for all $\bm{x}\in\gls*{roa}$, $\bm{w}\in\glsd*{wk}$
    \begin{subequations}\label{eq:iss_lyapunov}
    \begin{align}
			&\alpha_1\left(\norm{\bm{x}}\right) \le V\left(\bm{x}\right) \le \alpha_2\left(\norm{\bm{x}}\right),\label{eq:iss_lyapunov1} \\
			&V\left(\bm{f}\left(\bm{x}, \bm{w}\right)\right) - V\left(\bm{x}\right) \le - \alpha_3\left(\norm{\bm{x}}\right) + \gamma\left(\norm{\bm{w}}\right).\label{eq:iss_lyapunov2}
	\end{align}
	\end{subequations}
\end{definition}
\begin{lemma}[ISS via ISS-Lyap. funct. {\cite[Lemma~3.5]{jiang2001input}}]\label{lem:iss_lyap1}
	If system~\eqref{eq:sys_iss} admits a continuous ISS-Lyapunov function $V: \gls*{roa} \rightarrow \mathbb{R}_0^+$, then it is ISS with region of attraction $\gls*{roa}$.
\end{lemma}
Lemma~\ref{lem:iss_lyap1} states that ISS follows from the existence of a continuous ISS-Lyapunov function. Condition~\eqref{eq:iss_lyapunov2} depends on the disturbance $\bm{w}$.
The following lemma shows how this condition can be replaced with the undisturbed case $\bm{w}=\bm{0}$.
\begin{lemma}[Lipschitz ISS-Lyap. funct. {\cite[Lemma~22]{goulart2006optimization}}] \label{lem:iss_lyap_undisturbed}
	Let the set $\gls*{roa}$ contain the origin in its interior, be robust positive invariant for system~\eqref{eq:sys_iss}, and let $V: \gls*{roa} \rightarrow \mathbb{R}_0^+$ be an ISS-Lyapunov function for the undisturbed system~$\gls*{xk1} = \bm{f}\left(\gls*{xk}, \bm{0}\right)$.
	$V$ is an ISS-Lyapunov function for the original disturbed system~\eqref{eq:sys_iss}
	if 
	\begin{enumerate}
	    \item $V: \gls*{roa} \rightarrow \mathbb{R}_0^+$ is Lipschitz continuous on $\gls*{roa}$,
	    \item $\bm{f}: \gls*{roa} \times \glsd*{wk} \rightarrow \glsd*{xk}$ is Lipschitz continuous on $\gls*{roa} \times \glsd*{wk}$,\label{eq:f_lipschitz}
	    \item $\bm{f}\left(\bm{x}, \bm{w}\right) := \bm{f}_x\left(\bm{x}\right) + \bm{w}$, where $\bm{f}_x: \gls*{roa} \rightarrow \glsd*{xk}$ is continuous on $\gls*{roa}$.\label{eq:fx_lipschitz}
	\end{enumerate}
\end{lemma}
The following corollary is a direct result of the Lipschitz continuity of linear functions, and will be used to proof ISS of system~\eqref{eq:system1} under the proposed control law $\gls*{ulaw}\left(\gls*{xhk}\right)=\gls*{ulaw}\left(\gls*{xk}+\gls*{epsk}\right)$~\eqref{eq:controllaw},
\begin{equation}\label{eq:system_aut}
	\gls*{xk1} = 
	\bm{A} \gls*{xk} + \bm{B}\gls*{ulaw}\left(\gls*{xk}+\gls*{epsk}\right) + \gls*{Bd}\gls*{dk}.
\end{equation}
\begin{corollary}\label{cor:iss_lti}
    Consider a closed loop LTI system $\bm{f}\left(\gls*{xk}, \gls*{wk}\right) = \bm{A}\gls*{xk} + \bm{B}\gls*{ulaw}\left(\gls*{xk}\right)+ \gls*{wk}$ with additive disturbance $\gls*{wk}$ from a compact set $\glsd*{wk}$, and both $\gls*{roa}$ and $V: \gls*{roa} \rightarrow \mathbb{R}_0^+$ as in the hypothesis of Lemma~\ref{lem:iss_lyap_undisturbed}. 
    The closed loop $\bm{f}\left(\gls*{xk}, \gls*{wk}\right)$ is ISS with respect to disturbances $\gls*{wk}\in\glsd*{wk}$ if $\gls*{ulaw}(\gls*{xk})$ is Lipschitz continuous on $\gls*{X}_0$.
\end{corollary}
By Proposition~\ref{prop:lipschitz_ulaw}, $\gls*{ulaw}\left(\gls*{xhk}\right)$ is Lipschitz continuous in its variable $\gls*{xhk}$. However, the measurement noise $\gls*{epsk}$ inside $\gls*{xhk}=\gls*{xk}+\gls*{epsk}$~\eqref{eq:system2} renders $\gls*{ulaw}(\gls*{xk}+\gls*{epsk})$ discontinuous in $\gls*{xk}$. As a consequence, direct stability proofs for the evolution of the state $\gls*{xk}$ as in~\eqref{eq:system_aut} are challenging.
Therefore, as in~\cite{Roset.2008}, we first analyze stability for the equivalent evolution in terms of the state measurements~\eqref{eq:xmeasure_dynamics}
\begin{equation}\label{eq:system_aut_noisy}
	\gls*{xhk1} =  \bm{\hat{f}}(\gls*{xhk},\gls*{wk}) = \bm{A} \gls*{xhk} + \bm{B} \gls*{ulaw}\left(\gls*{xhk}\right) + \gls*{wk}.
\end{equation}
For the artificial noisy system~\eqref{eq:system_aut_noisy}, the control law $\gls*{ulaw}\left(\gls*{xhk}\right)$ is Lipschitz continuous in the state $\gls*{xhk}$ of the system and by Corollary~\ref{cor:iss_lti}, system~\eqref{eq:system_aut_noisy} is ISS with respect to the extended disturbance $\bm{w}$ if there is a Lipschitz continuous ISS-Lyapunov function for the uncertainty-free system
\begin{equation}\label{eq:system_aut_undisturbed}
\gls*{xhk1}=\bm{\hat{f}}(\gls*{xhk},\bm{0})=\bm{A} \gls*{xhk} + \bm{B} \gls*{ulaw}\left(\gls*{xhk}\right).
\end{equation}.
System~\eqref{eq:system_aut_undisturbed} coincides with the nominal dynamics~\eqref{eq:nominal_dynamics} assumed in the OCP. As a consequence, if the predictive control scheme was applied to system~\eqref{eq:system_aut_undisturbed}, the measured state at the next time step would be equal to the first predicted state (of the optimal state sequence),
\begin{equation}\label{eq:truePred}
    \hat{\bm{z}}_{1|k}^*=\gls*{xhk1}=\hat{\bm{z}}_{0|k+1}.
\end{equation}
This equality allows us to express (cost) functions of the successor state $\gls*{xhk1}$ in terms of the optimal solution of the OCP $\gls*{OCP}(\gls*{xhk})$, and thereby upper bound the descent of the cost function from time step $k$ to $k+1$.
In the following, we show that the optimal cost $J_L^*$ of the OCP,
\begin{equation}\label{eq:optimalcost}
	J_L^* \left(\hat{\bm{x}}_{k}\right) = \sum\limits_{l = 0}^{L-1} J_{\text{s}}\left(\hat{\bm{z}}_{l|k}^*, \bm{u}_{l|k}^*\right) + J_{\text{f}}\left(\hat{\bm{z}}_{L|k}^*\right),
\end{equation}
is an ISS-Lyapunov function for system~\eqref{eq:system_aut_undisturbed} with proposed control law $\gls*{ulaw}\left(\gls*{xhk}\right)=\bm{K}\gls*{xhk}+\bm{v}^{*}_{0|k}$~\eqref{eq:controllaw}, if the terminal cost function~\eqref{eq:terminalcost} is a Lyapunov function for system~\eqref{eq:system_aut_undisturbed} under pure state feedback $\gls*{ulaw}\left(\gls*{xhk}\right)=\bm{K}\gls*{xhk}$.
 
\begin{proposition}[Terminal cost function] \label{prop:termCost}
Let $\bm{P}$ be the optimal cost-to-go matrix~\eqref{eq:computeP} of the LQR-problem associated with the state feedback gain $\bm{K}$. The terminal cost function~\eqref{eq:terminalcost} is a Lyapunov function in the terminal set $\gls*{termSet}$~\eqref{eq:tightcons_terminal_noise} for the closed-loop system $\bm{x}_{k+1} = \left(\bm{A} + \bm{B} \bm{K}\right)\bm{x}_k$, such that
\begin{equation} \label{eq:terminaldecrescency}
    J_{\text{f}}\left(\bm{x}_{k+1}\right) - J_{\text{f}}\left(\gls*{xk}\right) = -J_{\text{s}}\left(\gls*{xk}, \bm{K}\gls*{xk}\right).
\end{equation}
\end{proposition}
The proof of Proposition~\eqref{prop:termCost} is standard in MPC and included in the supplementary material.

De Persis and Tesi derive the data-based computation of $\bm{K}$ in \cite[Theorem 4]{de2019formulas} by framing the LQR problem as a special 2-norm optimization problem. As a consequence, they do not touch on the solution of the associated algebraic Riccati equation, i.e., the matrix $\bm{P}$ defining the optimal cost-to-go $\bm{x}^T\bm{P}\bm{x}$.
\begin{proposition} \label{prop:compute_P}
    From the optimal LQR state-feedback gain $\bm{K}$, the matrix $\bm{P}$ can be computed as the solution of
    \begin{align}\label{eq:computeP}
    	& \underset{\bm{P}}{\text{minimize}}\,\text{trace}(\bm{P}) \\
    	\text{s.t. } & \tilde{\bm{W}}^{\top} \bm{X}^{\top}_{+} \bm{P}  \bm{X}_{+}\tilde{\bm{W}} - \bm{P} + \bm{Q}_{\bm{P}} \prec 0,\nonumber
    \end{align}
    where $\bm{Q}_{\bm{P}} = \bm{Q} + \bm{K}^{\top} \bm{R} \bm{K}$ and $\tilde{\bm{W}}$ is defined such that $\bm{U}\tilde{\bm{W}} =\bm{K}$ and $\bm{X}\tilde{\bm{W}} =\bm{I}_n$.
\end{proposition}
The proof of Proposition~\ref{prop:compute_P} is in the supplementary material.

We proof the next result with the help of an explicit candidate solution.
Given a feasible solution $\hat{\bm{Z}}_k^*, \bm{V}_k^*, \bm{\alpha}_k^*$ to the OCP~\eqref{eq:ocp} at time step $k$, a candidate solution for the time instant $k+1$ can be constructed by shifting the input trajectory $\bm{V}_k$ and appending $\bm{K}\hat{\bm{z}}_{L|k}^*$.
    Denote the optimal input sequence associated with $\bm{V}_k^*$ as
    $\bm{U}_k^* = (\bm{u}_{0|k}^*, \bm{u}_{1|k}^*, \dots, \bm{u}_{L-1|k}^*)$,
    ${\bm{u}_{i|k}^* = \bm{v}_{i|k}^* + \bm{K}\hat{\bm{z}}_{i|k}^*}$.
    A feasible, likely suboptimal, candidate solution for the next time step is
	$\tilde{\bm{U}}_{k+1} = (\bm{u}_{1|k}^*, \bm{u}_{2|k}^*, \dots, \bm{u}_{L-1|k}^*, \bm{K}\hat{\bm{z}}_{L|k}^*)$.
    Since predictions are exact~\eqref{eq:truePred} for the uncertainty-free system~\eqref{eq:system_aut_undisturbed}, we have $\gls*{xhk1}=\hat{\bm{z}}_{1|k}^*$ and the candidate state sequence is 
	$\tilde{\bm{Z}}_{k+1} = ( \hat{\bm{z}}_{1|k}^*, \hat{\bm{z}}_{2|k}^*, \dots, \hat{\bm{z}}_{L|k}^*, \tilde{\bm{z}}_{L|k+1} )$
    with ${\tilde{\bm{z}}_{L|k+1} = \left(\bm{A} + \bm{B}\bm{K}\right)\hat{\bm{z}}_{L|k}^*}=\left[\bm{H}_x\right]_{[Ln_\mathrm{x}+1,Ln_\mathrm{x}+n_\mathrm{x}]}\tilde{\bm{\alpha}}_{k+1}$
    and $\tilde{\bm{\alpha}}_{k+1} = \mat{\bm{H}_u\\ \left[\bm{H}_x\right]_{[1,n_\mathrm{x}]}}^{\dagger} \mat{\ul{\tilde{\bm{U}}}_{k+1}\\\hat{\bm{z}}_{1|k}^*}$.

    Based on Proposition~\ref{prop:termCost}, the difference between the terminal cost of the solution of the OCP at time step $k$ and the terminal cost of the explicit candidate solution at the next time step $k+1$ is upper bounded via 
    \begin{equation}\label{eq:upperbound}
        J_{\text{f}}\left(\tilde{\bm{z}}_{L|k+1}\right) - J_{\text{f}}\left(\hat{\bm{z}}_{L|k}^*\right)\le -J_{\text{s}}\left(\hat{\bm{z}}_{L|k}^*, \bm{K}\hat{\bm{z}}_{L|k}^*\right).
    \end{equation}
\begin{lemma}[ISS of the artificial system] \label{lemma:iss_noisy}
	The optimal cost $J_L^*$~\eqref{eq:optimalcost} is an ISS-Lyapunov function for system~\eqref{eq:system_aut_undisturbed}. Furthermore, $J_L^*$ is Lipschitz continuous on $\mathcal{C}_L^{\infty}$. As a consequence, system~\eqref{eq:system_aut_noisy} is ISS with region of attraction $\mathcal{C}_L^{\infty}$.
\end{lemma}
\begin{proof}
    $J_L^*$ is positive definite, defined on a domain that contains the origin with $J_L^*(\bm{0})=0$ and continuous by Proposition~\ref{prop:lipschitz_ulaw}. By~\cite[Lemma 4.3]{Khalil.2002}, there exist functions ${\alpha_1,\alpha_2 \in \mathscr{K}_{\infty}}$ such that condition~\eqref{eq:iss_lyapunov1} holds for $V=J_L^*$.
    Condition~\eqref{eq:iss_lyapunov2} is $J_L^*\left(\bm{\hat{f}}(\gls*{xhk},\bm{0})\right) - J_L^*\left(\gls*{xhk}\right) \le - \alpha\left(\norm{\gls*{xhk}}\right)$ for some function $\alpha_3\in\mathscr{K}_{\infty}$, 
    which we show with the help of the suboptimal candidate solution for the next time step $k+1$.
    The cost of the candidate solution, denoted by $\tilde{J}_L(\gls*{xhk1})$, is
    \begin{equation*}
        \sum\limits_{l = 1}^{L-1} J_{\text{s}}\left(\hat{\bm{z}}_{l|k}^*, \bm{u}_{l|k}^*\right) + J_{\text{s}}\left(\hat{\bm{z}}_{L|k}^*, \bm{K}\hat{\bm{z}}_{L|k}^*\right)+ J_{\text{f}}\left(\tilde{\bm{z}}_{L|k+1}\right).        
    \end{equation*}
    Since $J_L^* \left(\hat{\bm{x}}_{k+1}\right) \le \tilde{J}_L(\gls*{xhk1})$, the descent of the optimal cost function $J_L^* \left(\hat{\bm{x}}_{k+1}\right) - J_L^* \left(\hat{\bm{x}}_{k}\right)$ from time step $k$ to time step $k+1$ is upper bounded by $J_L^* (\hat{\bm{x}}_{k+1}) - J_L^* (\hat{\bm{x}}_{k}) \le \tilde{J}_L(\gls*{xhk1}) - J_L^* (\hat{\bm{x}}_{k})
		\underset{\eqref{eq:truePred}}{=} -J_{\text{s}}(\hat{\bm{x}}_{k}, \gls*{ulaw}(\hat{\bm{x}}_{k})) + J_{\text{s}}(\hat{\bm{z}}_{L|k}^*, \bm{K}\hat{\bm{z}}_{L|k}^*)+
    		J_{\text{f}}(\tilde{\bm{z}}_{L|k+1}) - J_{\text{f}}(\hat{\bm{z}}_{L|k}^*)
        \underset{\eqref{eq:upperbound}}{\le} -J_{\text{s}}(\hat{\bm{x}}_{k}, \gls*{ulaw}(\hat{\bm{x}}_{k}))$.
    Since $\gls*{ulaw}\left(\hat{\bm{x}}_{k}\right)$ is continuous in $\hat{\bm{x}}_{k}$ by Proposition~\ref{prop:lipschitz_ulaw}, and $J_{\text{s}}\left(\hat{\bm{x}}_k, \gls*{uk}\right)$~\eqref{eq:stagecost} is continuous in $\hat{\bm{x}}_{k}$ and $\gls*{uk}$, the composition $J_{\text{s}}\left(\hat{\bm{x}}_k, \gls*{ulaw}\left(\hat{\bm{x}}_{k}\right)\right)$ is continuous in $\hat{\bm{x}}_{k}$. Hence, there exists $\alpha_3\in\mathscr{K}_{\infty}$ such that ${\alpha_3\left(\norm{\hat{\bm{x}}_k}\right) \le J_{\text{s}}\left(\hat{\bm{x}}_k, \gls*{ulaw}\left(\hat{\bm{x}}_{k}\right)\right)}$ by ~\cite[Lemma 4.3]{Khalil.2002}.
    As a consequence, $J_L^* \left(\hat{\bm{x}}_{k+1}\right) - J_L^* \left(\hat{\bm{x}}_{k}\right) \le  - \alpha_3\left(\norm{\gls*{xhk}}\right)$.  
    With condition~\eqref{eq:iss_lyapunov1} and \eqref{eq:iss_lyapunov2} satisfied for $\bm{w}=0$, $J_L^*$ is an ISS-Lyapunov function for system~\eqref{eq:system_aut_undisturbed}.   
    
    By construction, $\mathcal{C}_L^{\infty}$~\eqref{eq:robustCIreachableset} is compact, contains the origin and is robust positive invariant for system~\eqref{eq:system_aut_noisy} with respect to all possible disturbances $\bm{w}_k\in\mathcal{W}$. Since the quadratic optimal cost $J_L^*\left(\cdot\right)$ is Lipschitz continuous on the compact region of attraction $\mathcal{C}_L^{\infty}$~\cite[Proposition 17]{goulart2006optimization}, system~\eqref{eq:system_aut_noisy} is ISS by Lemma~\ref{lem:iss_lyap_undisturbed}.
\end{proof}
Lemma~\ref{lemma:iss_noisy} shows that the artificial system~\eqref{eq:system_aut_noisy} is input-to-state stable under the proposed control law~\eqref{eq:controllaw}. ISS for the actual closed loop system~\eqref{eq:system_aut} follows from the measurement equation ${\gls*{xhk} = \gls*{xk} + \gls*{epsk}}$ and the definition of input-to-state stability, as in~\cite[Theorem 3.3]{Roset.2008}, based on properties of comparison functions.
\begin{theorem}[ISS of the closed loop system] \label{prop:iss_real}
System~\eqref{eq:system_aut} is ISS under the proposed control law~\eqref{eq:controllaw}, with region of attraction ${\glsd*{roa} = \mathcal{C}_L^{\infty} \ominus \glsd*{epsk}}$.
\end{theorem}
\begin{proof}
    The proof follows from~\cite[Theorem 3.3]{Roset.2008}, since both necessary Assumptions \cite[Assumpt. 3.1, 3.2]{Roset.2008} hold with
    $\bm{A}\bm{x}+\bm{B}\bm{u}$ Lipschitz continuous in $\bm{x}$, $\mathcal{W}$ compact and therefore contained in a closed ball of some radius $\lambda$, the artifical system~\eqref{eq:system_aut_noisy} ISS on $\glsd*{roa}$ w.r.t. the extended additive disturbance $\bm{w}\in\mathcal{W}$, and $\glsd*{roa}$ robust positive invariant w.r.t. the extended additive disturbance by construction.
    For completeness, a self contained proof is included in the supplementary material.
\end{proof}

In summary, we have shown that the proposed receding horizon OCP is recursively feasible and the resulting data-driven predictive controller~\eqref{eq:controllaw} renders the closed loop system ISS with respect to both additive disturbances and measurement noise.

\section{ON THE CASE OF INEXACT DATA}\label{sec:Discussion}
The availability of exact trajectory data is a standing assumption in this work (Assumption~\ref{assum:trajData}).
In essence, we take a \emph{certainty-equivalence} approach, i.e., we  assume that the data-driven representation captures the \emph{true} system, comparable to the case where system matrices $(\bm{A},\bm{B},\gls*{Bd})$ are known.
The proposed control scheme and the proof of its properties were nontrivial, since we investigated a realistic setting \emph{during} the control phase, with the system state influenced by unknown disturbances and only accessible via noisy measurements. 
Even though offline data can be averaged, filtered or similarly de-noised, in practice, data are still likely inexact and Assumption~\ref{assum:trajData} is likely violated.
As a consequence, the hypothesis of Lemma~\ref{lem:extfundamental} no longer holds. 

In that case, there are four consequences for the presented control scheme: 1) predictions within the OCP are inexact even for the nominal system; 2) the computed constraint sets are not as specified and may lose their declared properties; 3) the computed LQR gain $\bm{K}$ is not the actual LQR gain associated with specified $\bm{Q}$ and $\bm{R}$, and may even be non-stabilizing; 4) the computed solution of the algebraic Ricatti equation $\bm{P}$ is (likely) no longer associated with $\bm{K}$ via the Lyapunov equation. 

Problem 1) may be alleviated by a bound on the prediction error, which could be accounted for in a further constraint tightening. Specifying such a bound is a pressing open challenge in the data-driven control literature.
Problem 2) may be alleviated at the cost of conservatism, by guaranteeing either under- or overapproximation of the original sets, whichever is appropriate for the resulting properties in question.
Problem 3) is not as crucial, since only optimality is lost as long as $\bm{K}$ is still stabilizing. In~\cite{Dorfler.2023bridging}, D\"orfler et al. present a regularized data-driven LQR which leads to stabilizing controllers even when the signal-to-noise ratio is large.
Problem 4) has implications related to stability: the solution of the algebraic Ricatti equation $\bm{P}$ is used to guarantee that the terminal cost function is a Lyapunov function for the nominal system under control $\bm{u}=\bm{K}\bm{x}$. This in turn lets us bound the descent of the stage cost. A similiar bound may still be obtained without $\bm{P}$ being accurate.

In practice, deviations from the case of exact data can be reduced by implementing small regularizing adjustments, the benefits of which have been commonly observed in the literature on data-driven control with inexact data~\cite{coulson2019regularized,Dorfler.2023bridging,Doerfler.RegularizedLQR,Doerfler.RegularizedLQR}.
In particular, it is advantageous to penalize large norms of the decision variable $\bm{\alpha}$, because any noise inside the Hankel matrix~\eqref{eq:ocp_hankel} is amplified directly by $\bm{\alpha}$. In line with~\cite{Dorfler.2023bridging}, we propose to add the term $\lambda_{\alpha} \norm{\bm{\Pi}\bm{\alpha}_k}$ to the cost function~\eqref{eq:ocp_cost}, where ${\lambda_{\alpha} > 0}$ is a regularization parameter and $\bm{\Pi}$ is a matrix defined as
\begin{equation}
    \bm{\Pi} := \bm{I} - \bm{H}_{udx}^{\dagger}\bm{H}_{udx}, ~~ \bm{H}_{udx} := \mat{\bm{H}_u - \tilde{\bm{K}} \bm{H}_x\\ \bm{H}_d \\ \left[\bm{H}_{x}\right]_{[1,\,n]}}.
\end{equation}
With the above definition, $\norm{\bm{\Pi}\bm{\alpha}}$ is a measure for the distance of $\bm{\alpha}$ to the image of $\bm{H}_{udx}^{\dagger}$. For $\lambda_{\alpha}$ sufficiently large, the equality $\bm{\Pi}\bm{\alpha}_k = \bm{0}$ holds, which means that $\bm{\alpha}_k \in \text{image}\,\bm{H}_{udx}^{\dagger}$~\cite{Dorfler.2023bridging}.

The presented sets of admissible decision variables $\bm{\alpha}$, such as $\mathcal{A}_{\mu}$~\eqref{eq:MeasNoiseTube} of Lemma~\ref{lem:noise_tight} and $\mathcal{A}_e$ in \eqref{eq:alphaset_e} may be (partially) de-noised by enforcing
$\bm{\Pi} \bm{\alpha} = \bm{0}$ 
in the definition of each set.

Similarly, instead of directly computing state errors from disturbance samples with the pseudo-inverse in~\eqref{eq:stateErrorPred}, we propose to first obtain $\bm{\alpha}_e$ by solving
\begin{align}\label{eq:error_opt}
	\bm{\alpha}_e &= ~\argmin_{\tilde{\bm{\alpha}}_e} ~l_{\alpha}\left(\tilde{\bm{\alpha}}_e\right)\\	
	&\text{s.t. } ~  \mat{\bm{0}\\ \underline{\bm{d}}_{[0|k,\,L|k]} \\
	\bm{e}_0} = \mat{\bm{H}_u - \tilde{\bm{K}}\bm{H}_x \\ \bm{H}_d \\ \left[\bm{H}_{x}\right]_{[1,\,n]}} \tilde{\bm{\alpha}}_e,
\end{align}
for some convex cost function $l_{\alpha}(\cdot):~ \mathbb{R}^{N-L} \rightarrow \mathbb{R}_{\ge 0}$ that penalizes large $\bm{\alpha}_{e}$.

A regularized version of~\cite[Theorem 4]{de2019formulas} was proposed in \cite{Doerfler.RegularizedLQR}. Here, the optimization problem is extended by the term $\lambda_{\alpha} \norm{\bm{\Pi}\bm{W}}$ in the cost function (or alternatively, by adding the constraint $\bm{\Pi}\bm{W} = \bm{0}$).
\section{SIMULATION EXAMPLE}\label{sec:Sim}
In this section, we present a simulation example and demonstrate the effectiveness of the proposed control algorithm.
Consider the double-mass-spring-damper system, shown in Fig. \ref{fig:doublemass}. 
\begin{figure}
    \centering
    \includegraphics[width=0.45\textwidth]{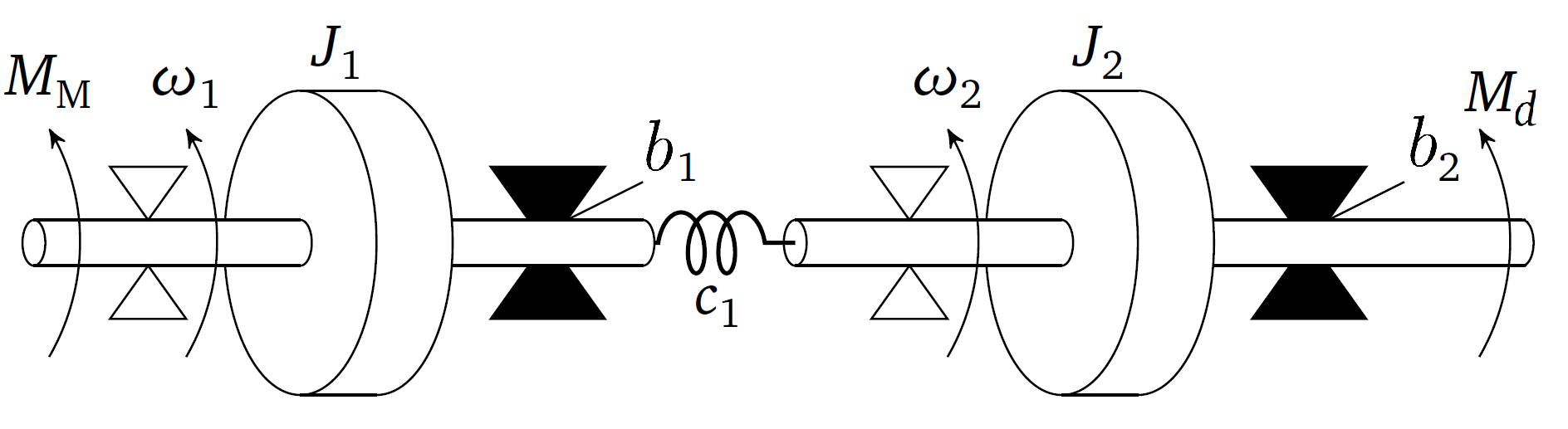}
    \caption{Double-mass-spring-damper model used for the simulation. The states are angles and angular velocities of the masses, the input is the torque on mass 1, and the disturbance is the torque on mass 2.}
    \label{fig:doublemass}
\end{figure}
The discrete-time dynamics of the system read
\begin{multline}\label{eq:evalsys_k}
    \bm{x}_{k+1} = \mat{0.952 & 0.048 & 0.094 & 0.002 \\
    	0.048 & 0.952 & 0.002 & 0.094 \\
    	-0.920 & 0.920 & 0.859 & 0.046 \\
    	0.920 & -0.920 & 0.046 & 0.858} \bm{x}_k \\ + 
    	\mat{0.048 \\ 0.001 \\ 0.936 \\ 0.016} u_k + 
    	\mat{0.001 \\ 0.048 \\ 0.016 \\ 0.94} d_k,
\end{multline}
following discretization of the continuous-time dynamics (see supplementary material)
with sampling time $\Delta t=\SI{0.1}{\second}$. The state vector $\bm{x} = \mat{\theta_1 & \theta_2 & \omega_1 & \omega_2}$ consists of the respective angles and angular velocities of the two masses. The input $u = M_{\textnormal{M}}$ is the torque on mass 1, the disturbance $d = M_\textnormal{d}$ is the torque on mass 2.
The goal is to stabilize the origin while minimizing cumulative stage costs $\sum_k\, \bm{x}_k^\top\bm{Q}\bm{x}_k+u_k R u_k$ with ${\bm{Q}=\text{diag}(10,10,1,1)}$, $R = 0.1$, and respecting input constraints $|u|\le1\si{.\newton\meter}$ and component-wise state constraints $-\bm{x}_{\text{max}} \le \bm{x} \le \bm{x}_{\text{max}}$, with $\bm{x}_{\text{max}} = [2\pi \,~ 2\pi \,~ 0.5\pi\si{.\second^{-1}} \,~ 0.5\pi\si{.\second^{-1}}]^{\top}$.
The disturbance $\gls*{dk}$ acting on mass 2 is sampled randomly from a zero-mean normal distribution with variance $\ol{d} = 0.1\si{.\newton\meter}$, truncated at the bounds of the interval $[-\ol{d},\,\ol{d}]$.
Online state measurements are corrupted by measurement noise $\bm{\mu}_k$, which is sampled from a zero-mean normal distribution with variance $\ol{\mu}= 0.015$, truncated such that $\norminf{\bm{\mu}} \le \ol{\mu}$.

\subsection{Exact offline data}
For the proposed control algorithm, the model \eqref{eq:evalsys_k} is assumed to be unknown.
To retrieve trajectory data as in Assumption~\ref{assum:trajData}, we simulate the system in open-loop for $50$ time-steps, with admissible inputs chosen at random.
With the trajectory data, we generate a persistently excited nominal state trajectory via~\eqref{eq:alphaz}, which lets us compute the LQR feedback gain $\bm{K}$ 
based on~\cite[Theorem 4]{de2019formulas} and the specified $\bm{Q}$, $R$. 
By Proposition~\ref{prop:compute_P}, we retrieve the associated matrix $\bm{P}$
for the definition of the terminal cost function.
For the probabilistic constraint tightening, we use 2924 disturbance measurements (Assumption~\ref{assum:DisturbanceBounds}a)
to solve the probabilistic optimization problems~\eqref{eq:stateprobcons_eta}, \eqref{eq:inputprobcons_mu}, with confidence $\beta = 0.99$ for a risk parameter $p \in \left[0.88,\,0.92\right]$.
The resulting constraint sets are further tightened to account for the bounded measurement noise (see Section~\ref{sec:tightcons}), for which we assume the bound $\ol{\mu}$ to be known (Assumption~\ref{assum:DisturbanceBounds}b).
We employ a Monte-Carlo simulation of 100 runs, each with different realizations of the online measurement noise and additive disturbance, a length of $50$ time steps, and fixed initial state $\bm{x}_0 = \mat{0.5\pi & 0.5\pi & 0\si{.\second^{-1}} & 0\si{.\second^{-1}}}^{\top}$.

All simulations are carried out in MATLAB, with polyhedral constraints specified based on MPT3~\cite{MPT3}, linear matrix inequalities solved with CVX~\cite{cvx}
, and the OCP~\eqref{eq:ocp} solved by MATLAB's \texttt{quadprog}.
With a prediction horizon set to $L = 10$, the mean computation time for the solution of the OCP was $1.66\si{.\milli\second}$ on an Intel I5-13600K.
The resulting trajectories are shown in Fig.~\ref{fig:trajNoisefreeMC}.
\begin{figure}
    \includegraphics[width=\columnwidth]{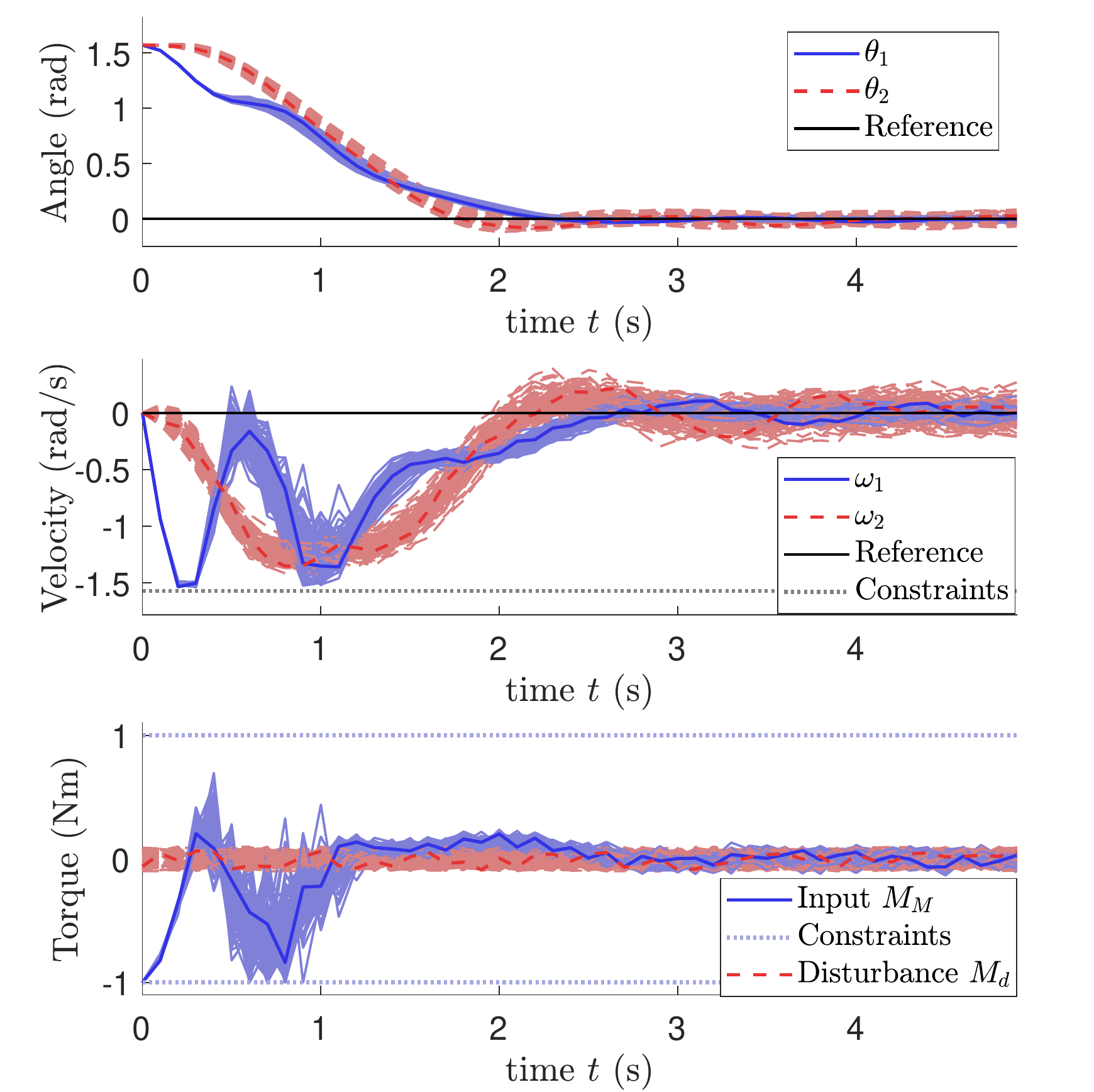} 
    \vspace{-15pt}
	\caption{Trajectories of 100 simulations for different realizations of online disturbance and measurement noise. The trajectory with median cost is highlighted.}
	\label{fig:trajNoisefreeMC}
\end{figure}
In all scenarios, the OCP remained feasible for all time steps and the proposed controller stabilized the system around the origin while respecting input and state constraints.
\subsection{Inexact offline data}
Although inexact data is not explicitly considered in the control design beyond the regularization techniques presented in Section~\ref{sec:Discussion}, we demonstrate the applicability of our proposed control algorithm in a more realistic setting.
To that end, we perturb all measured states and disturbances (including those used for the probabilistic constraint tightening) by measurement noise $\bm{\mu}$ as defined above. In other words, we now consider the same measurement noise to be present both online and offline, and also perturb disturbance samples, emulating a prior disturbance estimation procedure. 
In order to investigate the effect of noise inside the offline data, we sample 5 different realizations of offline noise and use it to perturb the same state and disturbance measurements. For each of the 5 sets of data, we use the same procedure as described above to compute tightened constraint sets, controller gain $\bm{K}$, and cost-to-go $\bm{P}$, but employ simple robustifications as discussed in Section~\ref{sec:Discussion}.
In the cost function of all 5 resulting OCPs, we add the term $5\norm{\bm{\Pi}\bm{\alpha}}_2^2$.

Trajectories for the 5 different resulting predictive controllers are shown in Fig.~\ref{fig:trajNoisy}. While the control performance varies significantly, all 5 predictive controllers stabilized the system around the origin while respecting state constraints.
\begin{figure}
    \includegraphics[width=\columnwidth]{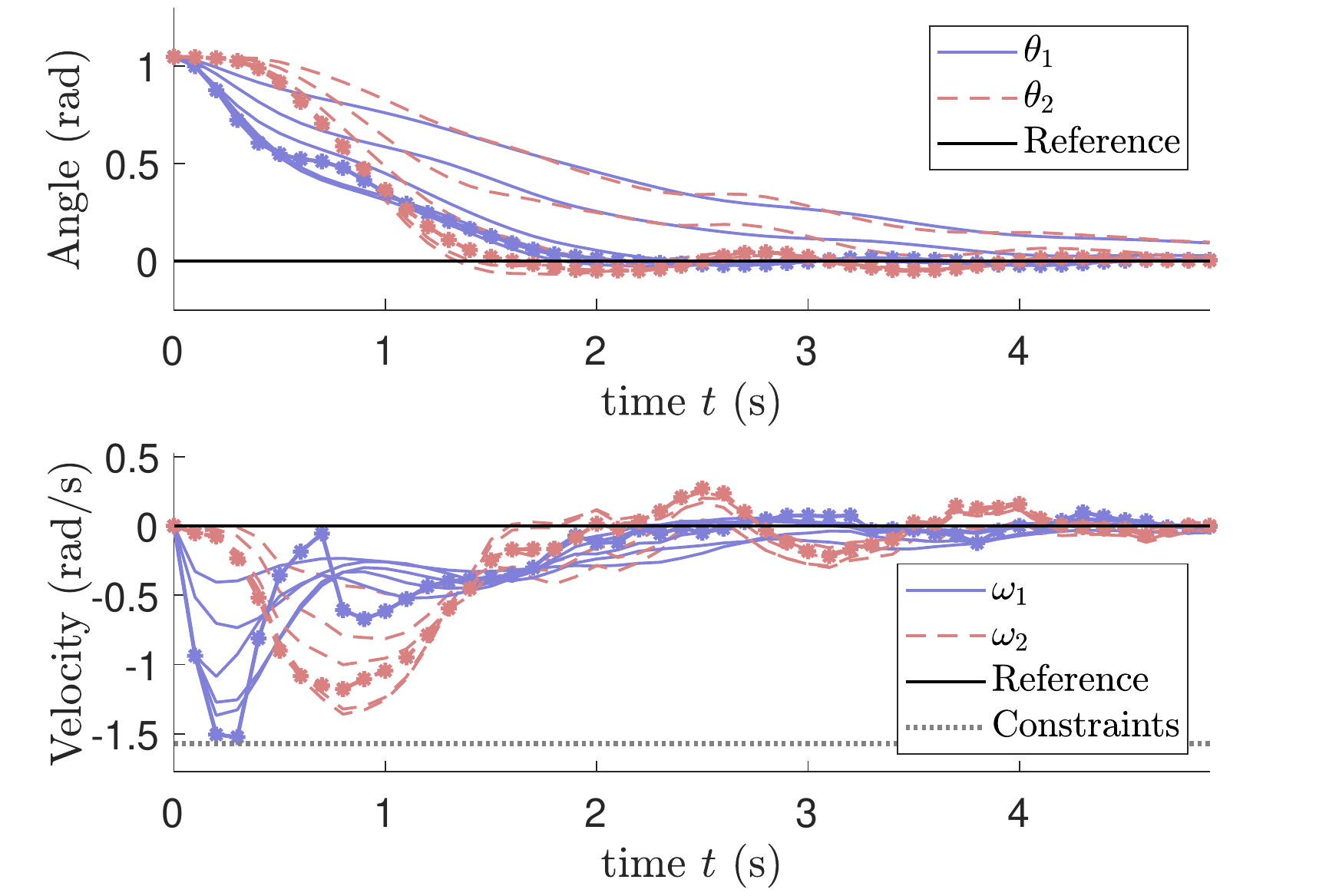}
    \vspace{-17pt}
	\caption{Median Trajectories based on offline state and disturbance data perturbed by measurement noise, each corresponding to different realizations of offline noise. Realizations of online disturbance and measurement noise are fixed. For comparison, the trajectory with exact offline data is shown and marked with $*$.}
	\label{fig:trajNoisy}
	\vspace{-8pt}
\end{figure}
For all 5 resulting data-driven predictive controllers, we again simulate 100 different realizations of online measurement noise and disturbance each. The resulting mean and standard deviation of the total costs can be seen in Table~\ref{table}. Table~\ref{table} also includes a comparison with the predictive controller based on unperturbed, noise-free, offline data.
\begin{table}
    \label{table}
    \tablefont
    \caption{Performance Comparison of Exact and Noisy (Inexact) State and Disturbance Data}
    \begin{tabular*}{20.25pc}{c|c|c|c|c}
        $[\bm{x}_0]_{[1,2]}$ & Data  & OCP infeasible & Mean Cost & STD Cost\\
        \hline
        $[\pi/2\,\pi/2]^\top$ & Exact & $0$ & $410.9$ & $9.8$\\
        $[\pi/2 \, \pi/2]^\top$ & Inexact & $14.7\%$ & $481.0239$ & $87.0$ \\
        $[\pi/3 \, \pi/3]^\top$ & Exact & $0$ & $153.6$ & $4.0$\\
        $[\pi/3 \, \pi/3]^\top$ & Inexact & $0.004\%$ & $195.5$ & $55.1$
    \end{tabular*}
    \vspace{-12pt}
\end{table}
Due to inaccurately computed constraint sets, the OCP was infeasible for $14.7\%$ of all 25000 time steps. In cases where the OCP was infeasible, we used pure state feedback $u_k=\bm{K}\gls*{xhk}$ (restricted to the specified input constraints) as a backup. As a consequence, the comparison of costs is not completely fair, especially since the application of pure state feedback lead to frequent state constraint violations (more than $14\%$ of time steps, compared to $0$ violations for the noise-free controller).
For a fair comparison, we chose a different initial state that is closer to the origin and reran all simulations. The results can be seen in the last two rows of the table. 
Here, the OCP was infeasible exactly once (out of 500 scenarios a 50 time steps, i.e., 25000 OCPs). 

With noise, mean costs increased by $27.3\%$, whereas median costs increased by only $7.5\%$ from $153.07$ to $164.47$. 
We observed that in our setting, the effect of noisy data is often small, and rarely very large. In hindsight, this further motivates a stochastic, as opposed to robust, future road to data-driven control.

\section{Conclusion}\label{sec:conclusion}
This work presented a novel data-driven stochastic predictive control scheme for the control of constrained LTI systems subject to stochastic disturbances and measurement noise. Our approach was motivated by the need to develop a lightweight and efficient predictive control strategy that leverages the deterministic fundamental lemma while also taking into account the probabilistic information of disturbances for performance-oriented control.
The key idea of the paper was to provide a data-driven
formulation of robust and stochastic tubes that leads to a recursively
feasible receding horizon OCP and a predictive controller that
renders the system input-to-state stable with respect to both
disturbance and measurement noise.
On the one hand, we provided a lightweight deterministic receding horizon OCP that allows for the data-driven control of stochastic systems without the need of a model.
On the other hand, we provided a data-driven tube-based formulation that allows for translation of state-of-the-art stochastic and robust tube-based MPC results into the data-driven domain.
\section*{Supplementaries / Appendix}
\subsection{On the sampling-based solution of Chance Constrained optimization}
If the chance constraint~\eqref{eq:ChanceOptimConstraint} is approximated by a finite number of samples, its satisfaction cannot be guaranteed with certainty, but parameters $p_\mathrm{min},p_\mathrm{max}$ specify a range $p\in [p_\mathrm{min},p_\mathrm{max}]$ for the risk parameter $p$ in~\eqref{eq:ChanceOptimConstraint} with confidence $1-\beta$ if
\begin{subequations}\label{eq:samplingNr}
    \begin{align}
        (1-p_\mathrm{max})\gls*{numSamp} - 1 + \sqrt{3(1-p_\mathrm{max})\gls*{numSamp}\ln{\frac{2}{\beta}}} &\leq \gls*{numDiscSamp},\\
        (1-p_\mathrm{min})\gls*{numSamp} - \sqrt{2(1-p_\mathrm{min})\gls*{numSamp}\ln{\frac{1}{\beta}}} &\geq \gls*{numDiscSamp},
    \end{align}
\end{subequations}
where $(1-p_\mathrm{min})\gls*{numSamp}\ge \gls*{numDiscSamp}$~\cite{Campi.2011}.
Given a fixed number of available samples $\gls*{numSamp}$, \eqref{eq:samplingNr} allows for a choice of the risk parameter range $p_\mathrm{min},p_\mathrm{max}$ and confidence $1-\beta$.

The approximated chance constrained optimization problem is then solved for each step $l\in\mathbb{N}_{[1,L]}$ of the prediction horizon as follows:
First, all state error samples $\bm{e}^{(i)}_{l}$ are computed from the available disturbance sequences $\bm{d}^{(i)}_{[0,\,l-1]}$, $i=0,\ldots \gls*{numSamp}$. Second, the error samples are multiplied from the left by $\bm{G}_x$ and the ordered set $\mathcal{G}_l=\{\bm{G}_x\bm{e}\,\mid\, \bm{e}=\bm{e}_{l}~\text{via \eqref{eq:error_dynamics} for some}~\bm{d}^{(i)}_{[0,\,l-1]} \}$ is constructed.
Last, the tightening parameter $\bm{\eta}_l$~\eqref{eq:DetEta} is computed component-wise based on
\begin{equation}
    [\bm{\eta}_l]_i = [\bm{g_x}]_i - q_l,
\end{equation}
where $q_l$ is the $(1-\frac{\gls*{numDiscSamp}}{\gls*{numSamp}})$-quantile of the ordered set $\mathcal{G}_l$.
\subsection{Computing LQR state feedback from data}
\begin{theorem}[LQR From Data{{~\cite[Theorem 4]{de2019formulas}}}] \label{th:stabilizing_K}
Consider system~(2) with $\gls*{dk}=\bm{0}$ and data persistently exciting of order $L+n_{\mathrm{x}}+1$ as in Assumption~1. Define the input and state data matrices ${\bm{U}:= \mat{\gls*{udata}_0 & \cdots & \gls*{udata}_{N-1}}}$, ${\bm{X} := \mat{\gls*{xdata}_0 & \cdots & \gls*{xdata}_{N-1}}}$, ${\bm{X}^{+} :=\mat{\gls*{xdata}_1 & \cdots & \gls*{xdata}_{N}}}$.
The optimal LQR state-feedback gain $\bm{K}$ for system~(2a) can be computed as $\bm{K} = \bm{U} \bm{W} \left(\bm{X} \bm{W}\right)^{-1}$, where $\bm{W} \in \mathbb{R}^{N \times n_{\mathrm{x}}}$ is the optimizer of
\begin{subequations}
	\begin{align*}
		& \underset{\bm{W}, \bm{V} }{\text{minimize}}~\text{trace}\left(\bm{Q} \bm{X} \bm{W}\right) + \text{trace}\left(\bm{V}\right)\\	
	    \text{s.t. }& \mat{\bm{V} & \bm{R}^{1/2} \bm{U} \bm{W} \\ \bm{W}^{\top} \bm{U}^{\top} \bm{R}^{1/2} & \bm{X} \bm{W}} \succeq 0,\\
		& \mat{\bm{X} \bm{W} - \bm{I}_{n_{\mathrm{x}}} & \bm{X}_{+} \bm{W} \\ \bm{W}^{\top} \bm{X}^\top_{+} & \bm{X} \bm{W}} \succeq 0,
	\end{align*}
\end{subequations}
where $\bm{Q}$, $\bm{R}$ are weighting matrices for the LQR problem.
\end{theorem}

\subsection{Proof of Proposition 4} 
First, let us recall that $\bm{P} \succ 0$ is the solution of the Discrete Algebraic Ricatti Equation (DARE)
\begin{equation} \label{eq:DARE}
    \bm{A}^{\top} \bm{P} \bm{A} - \bm{P} - \bm{A}^{\top} \bm{P} \bm{B} \left( \bm{R} + \bm{B}^{\top} \bm{P} \bm{B} \right) \bm{B}^{\top} \bm{P} \bm{A} = -\bm{Q},
\end{equation}
and the corresponding LQR state feedback gain $\bm{K}$ is
\begin{equation} \label{eq:LQRgain}
    \bm{K} = - \left( \bm{R} + \bm{B}^{\top} \bm{P} \bm{B} \right)^{-1} \bm{B}^{\top} \bm{P} \bm{A}.
\end{equation}
By combining \eqref{eq:DARE} and \eqref{eq:LQRgain}, we retrieve the Lyapunov equation
\begin{equation} \label{eq:DARElyap}
    \left( \bm{A} + \bm{B} \bm{K} \right)^{\top} \bm{P}  \left( \bm{A} + \bm{B} \bm{K} \right) - \bm{P} = - \bm{Q}_{\bm{P}},
\end{equation}
with $\bm{Q}_{\bm{P}} = \bm{Q} + \bm{K}^{\top} \bm{R} \bm{K} \succ 0$.
Recall the definition of the terminal cost $J_{\text{f}}(\bm{x})=\bm{x}^{\top}\bm{P}\bm{x}$~(34) and the stage cost $J_{\text{s}}\left(\gls*{xk}, \bm{u}\right) = \gls*{xk}^{\top} \bm{Q} \gls*{xk} + \gls*{uk}^{\top} \bm{R} \gls*{uk}$~(33).
With ${\gls*{xk1} = \left(\bm{A} + \bm{B} \bm{K}\right)\gls*{xk}}$, the descent of the terminal cost $\Delta J_{\text{f}}(\gls*{xk})= J_{\text{f}}\left(\gls*{xk1}\right) - J_{\text{f}}\left(\gls*{xk}\right)$ can be written as
\begin{equation*}
    \Delta J_{\text{f}}(\gls*{xk}) = \gls*{xk}^{\top} \left( \left( \bm{A} + \bm{B} \bm{K} \right)^{\top} \bm{P}  \left( \bm{A} + \bm{B} \bm{K} \right) - \bm{P} \right) \gls*{xk}.
\end{equation*}
With $\gls*{uk}=\bm{K}\gls*{xk}$, $J_{\text{s}}\left(\gls*{xk}, \bm{K}\gls*{xk}\right) = \gls*{xk}^{\top} \bm{Q}_{\bm{P}} \gls*{xk}$ and the equality
$J_{\text{f}}\left(\bm{x}_{k+1}\right) - J_{\text{f}}\left(\gls*{xk}\right) = -J_{\text{s}}\left(\gls*{xk}, \bm{K}\gls*{xk}\right)$
from Proposition~4 is satisfied due to the Lyapunov equation~\eqref{eq:DARElyap}.

\subsection{Proof of Proposition 5} 
    Given an LQR gain $\bm{K}$, the solution $\bm{P}$ of the Lyapunov equation~\eqref{eq:DARElyap} can be retrieved as the minimizer of the convex program
    \begin{align*}
    	& \underset{\bm{P}}{\text{minimize}}~\text{trace}\left(\bm{P}\right) \\
    	\text{s.t. } & \left( \bm{A} + \bm{B} \bm{K} \right)^{\top} \bm{P}  \left( \bm{A} + \bm{B} \bm{K} \right) - \bm{P} + \bm{Q}_{\bm{P}} \prec 0.
    \end{align*}
    Based on \cite[Theorem~2]{de2019formulas}, the closed-loop system dynamics $\bm{A} + \bm{B} \bm{K}$ can be substituted with the equivalent data-driven representation $\bm{X}_{+}\tilde{\bm{W}}$ for some $\tilde{\bm{W}}$ that satisfies $\bm{U}\tilde{\bm{W}} =\bm{K}$ and $\bm{X}\tilde{\bm{W}} =\bm{I}_{n_{\mathrm{x}}}$.

\subsection{Self-contained proof of Theorem 2} 
Since ${\bm{x}_0 \in \mathcal{C}_L^{\infty} \ominus \glsd*{epsk}}$ and ${\gls*{xhk} = \gls*{xk} + \gls*{epsk}}$, we have ${\hat{\bm{x}}_0 \in \mathcal{C}_L^{\infty}}$. 
Denote with $\hat{\bm{\phi}}\left(k,\hat{\bm{x}}_0, \bm{w}_{[0,k]}\right)$ the trajectory of~(54) for initial state $\hat{\bm{x}}_0$ and disturbance sequence $\bm{w}_{[0,k]}$, so that $\gls*{xhk}=\hat{\bm{\phi}}\left(k,\hat{\bm{x}}_0, \bm{w}_{[0,k]}\right)$.
Following Lemma~7, by definition of input-to-state stability, we can bound $\hat{\bm{\phi}}$ with functions $\hat{\beta}\in\mathscr{KL}$, $\hat{\gamma}\in\mathscr{K}$, such that
\begin{equation*} \label{eq:sol_noisy}
\norm{\hat{\bm{\phi}}\left(k,\hat{\bm{x}}_0, \bm{w}_{[0,k]}\right)} \le \hat{\beta}\left(\norm{\hat{\bm{x}}_0},k\right) + \hat{\gamma}\left(\norm{\bm{w}_{[0,k-1]}}_\infty\right)
\end{equation*}
for all initial states ${\hat{\bm{x}}_0 \in \mathcal{C}_L^{\infty}}$ and disturbance sequences $\bm{w}_{[0,k]}$, $\bm{w}_i\in\glsd*{wk}$.
Now denote the trajectory of the actual closed loop system~\eqref{eq:system_aut}
subject to measurement noise sequence $\gls*{noise}_{[0,k]}$ as $\gls*{xk}={\bm{\phi}\left(k,\bm{x}_0, \bm{d}_{[0,k]}, \gls*{noise}_{[0,k]}\right)}$.
From the above, a bound for the actual trajectory follows as
\begin{align*}
    &\norm{\bm{\phi}\left(k,\bm{x}_0, \bm{d}_{[0,k]}, \gls*{noise}_{[0,k]}\right)} \le \norm{\hat{\bm{\phi}}\left(k,\hat{\bm{x}}_0, \bm{w}_{[0,k]}\right)} + \norm{\gls*{noise}_0} \\	&\le \hat{\beta}\left(\norm{\hat{\bm{x}}_0},k\right) + \hat{\gamma}\left(\norm{\bm{w}_{[0,k]}}_\infty\right) + \norm{\gls*{noise}_0}.
\end{align*}
From here, condition~(8) for ISS follows by properties of the comparison functions~\cite{Kellett.2014}. 
Since $\hat{\beta}\in\mathscr{KL}$, there exist functions ${\beta \in \mathscr{KL}}$, ${\gamma_1 \in \mathscr{K}}$ such that
\begin{align*}
    \hat{\beta}\left(\norm{\hat{\bm{x}}_0},k\right) &\le \hat{\beta}\left(\norm{\bm{x}_0} + \norm{\gls*{noise}_0},k\right)\\
    &\le \beta\left(\norm{\bm{x}_0},k\right) + \gamma_1\left(\norm{\gls*{noise}_0}\right).
\end{align*}
Since $\hat{\gamma}\in\mathscr{K}$, there exist functions $\gamma_2, \gamma_3, \gamma_4 \in  \mathscr{K}$ such that
\begin{multline*}
\hat{\gamma}\left(\norm{\bm{w}_{[0,k]}}_\infty\right) =
\hat{\gamma}\left(\norm{\gls*{Bd} \bm{d}_{[0,k]} - \bm{A}\gls*{noise}_{[0,k]} + \gls*{noise}_{[0,k+1]}}_\infty\right) \\
\le \hat{\gamma}\left(\norm{\gls*{Bd} \bm{d}_{[0,k]}}_\infty +
\norm{\bm{A}\gls*{noise}_{[0,k]}}_\infty+\norm{\gls*{noise}_{[0,k+1]}}_\infty\right)\\
\le \gamma_2\left(\norm{\gls*{Bd} \bm{d}_{[0,k]}}_\infty\right)+\gamma_3\left(\norm{\bm{A}\gls*{noise}_{[0,k]}}_\infty\right)
 + \gamma_4\left(\norm{\gls*{noise}_{[0,k+1]}}_\infty\right).
\end{multline*}
Since there is some function $\gamma \in \mathscr{K}$ for which
\begin{multline*}	            
    \norm{\gls*{noise}_0}+\gamma_1\left(\norm{\gls*{noise}_0}\right) + \gamma_2\left(\norm{\gls*{Bd} \bm{d}_{[0,k]}}_\infty\right)
    +\gamma_3\left(\norm{\bm{A}\gls*{noise}_{[0,k]}}_\infty\right) \\
     + \gamma_4\left(\norm{\gls*{noise}_{[0,k+1]}}_\infty\right) 
    \le \gamma\left(\norm{\bm{d}_{[0,k]}}_\infty + \norm{\gls*{noise}_{[0,k]}}_\infty\right),
\end{multline*}
the actual trajectory is bounded as
\begin{multline*}
\norm{\bm{\phi}\left(k,\bm{x}_0, \bm{d}_{[0,k]}, \gls*{noise}_{[0,k]}\right)} \le \beta\left(\norm{\bm{x}_0},k\right) +\\
+ \gamma\left(\norm{\bm{d}_{[0,k]}}_\infty+\norm{\gls*{noise}_{[0,k]}}_\infty\right)
\end{multline*}
and system~\eqref{eq:system_aut} is ISS by Definition~4.
\subsection{Double-mass-spring system}
The continuous-time dynamics of the double-mass-spring system are given as
\begin{equation}\label{sys:cont_doublemass}
    \dot{\bm{x}} = \mat{0 & 0 & 1 & 0 \\
    	0 & 0 & 0 & 1 \\
    	-\frac{c_1}{J_1} & \frac{c_1}{J_1} & -\frac{b_1}{J_1} & 0 \\
    	\frac{c_1}{J_2} & -\frac{c_1}{J_2} & 0 & -\frac{b_2}{J_2}} \bm{x} + 
    	\mat{0 \\ 0 \\ \frac{1}{J_1} \\ 0} u+ 
    	\mat{0\\ 0 \\ 0 \\ \frac{1}{J_2}} d.
\end{equation}
with moments of inertia $J_1=J_2=\SI{0.1}{\kilo\gram^2}$, damping ratio $b_1=b_2=\SI{0.1}{\newton\meter / (rad/\second)}$ and spring constant $c_1=\SI{1}{\newton/rad}$.

\bibliography{mybib}

\begin{IEEEbiography}[{\includegraphics[width=2.5cm,height=3.2cm,clip,keepaspectratio]{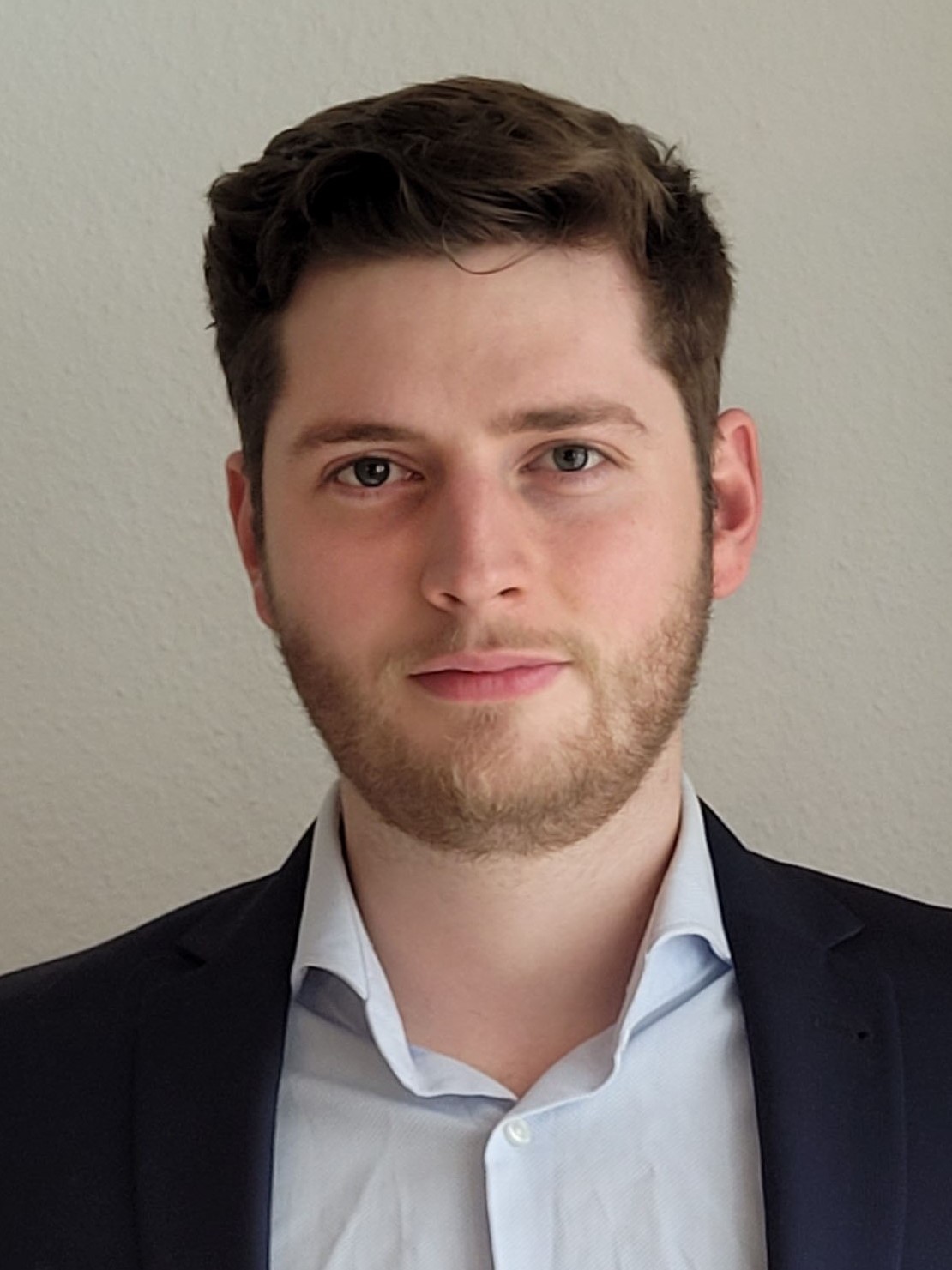}}]{Sebastian Kerz}{\space}(Graduate Student Member, IEEE) received his M.Sc. Degree in electrical engineering and information technology from Technical University Darmstadt, Germany, in 2019. During his studies, he was a visiting student at the University of Toronto in 2017. He is currently a Research Associate and Ph.D. student with the Chair of Automatic Control Engineering at the Technical University of Munich. His research interests include data-driven control, optimal control, and the intersection between control theory and machine learning.
\end{IEEEbiography}

\begin{IEEEbiography}[{\includegraphics[width=1in,height=1.25in,clip,keepaspectratio]{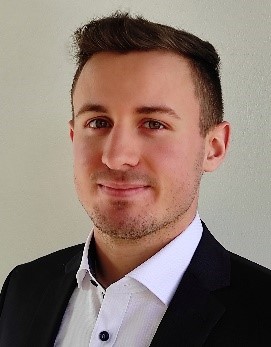}}]{Johannes Teutsch}{\space} received his B.Sc. and M.S. degree in electrical engineering and information technology from the Technical University of Munich (TUM), Germany in 2019 and 2021, respectively. He is currently a Research Associate and Ph.D. student with the Chair of Automatic Control Engineering at the Technical University of Munich. His research revolves around (model) predictive control, with a focus on data-driven predictive control in uncertain environments.
\end{IEEEbiography}

\begin{IEEEbiography}[{\includegraphics[width=1in,height=1.25in,clip,keepaspectratio]{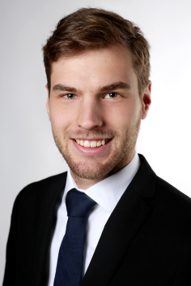}}]{Tim Br\"udigam}{\space} (Member, IEEE) received his B.Sc. and M.Sc. degree in electrical engineering from the Technical University of Munich (TUM), Germany in 2014 and 2017, respectively. In 2022, he obtained his Ph.D. degree from TUM.
During his studies, he was a research scholar at the California Institute of Technology in 2016. He joined the Chair of Automatic Control Engineering at TUM as a research associate in 2017 and was a visiting researcher at the University of California, Berkeley in 2021. His main research interest lies in advancing Model Predictive Control (MPC), especially stochastic MPC, with possible application in automated driving.
\end{IEEEbiography}

\begin{IEEEbiography}[{\includegraphics[width=1in,height=1.25in,clip,keepaspectratio]{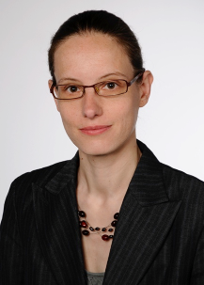}}]{Marion (nee Sobotka) Leibold}{\space} received the
Diploma degree in applied mathematics (2002) from
Technical University of Munich, Germany. Further
she received the Ph.D. (2007) and Habilitation
(2020) in control theory from Technical University
of Munich, Germany. She is currently a Senior Researcher with the Faculty of Electrical Engineering
and Information Technology, Institute of Automatic
Control Engineering, Technical University of Munich. Her research interests include optimal control
and nonlinear control theory, and the applications to
robotics and automated driving.
\end{IEEEbiography}

\begin{IEEEbiography}[{\includegraphics[width=1in,height=1.25in,clip,keepaspectratio]{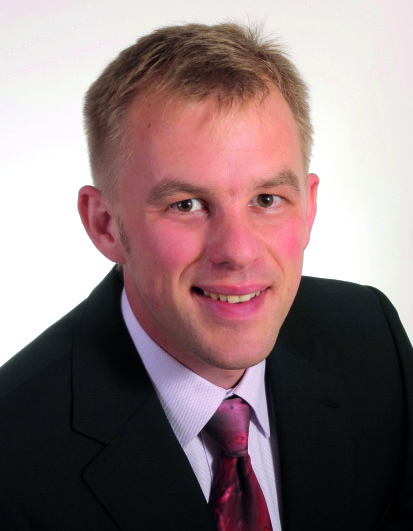}}]{Dirk Wollherr}{\space}(Senior Member, IEEE) received the Dipl.-Ing. (2000), Dr.-
Ing. (2005) and Habilitation (2013) degree in electrical engineering from Technical University of
Munich, Germany. He is a Senior Researcher in
robotics, control and cognitive systems. His research interests include automatic control, robotics
autonomous mobile robots, human-robot interaction
and socially aware collaboration and joint action. He
is on the management boards of several conferences
and Associate Editor of the IEEE Transactions on
Mechatronics.
\end{IEEEbiography}
\end{document}